\documentclass[12pt]{article}

\usepackage[utf8]{inputenc}

\usepackage{geometry} 
\usepackage[usenames,dvipsnames,svgnames,table]{xcolor}
\usepackage{graphicx} % support the \includegraphics command and options
\usepackage{amsmath}
\usepackage{amsthm}
\usepackage{amsfonts}
\usepackage{mathrsfs}
\usepackage{listings}
\usepackage{tabto}
\usepackage{array} 
\usepackage{verbatim}
\usepackage{nicefrac}

\usepackage{enumitem}
\usepackage{booktabs} 
\usepackage{tabularx}

\usepackage{algorithmic}

\usepackage[colorlinks, citecolor=black, urlcolor=black, linkcolor=black]{hyperref}

\usepackage[nohyperlinks, nolist]{acronym} %Abkürzungen

\usepackage{tikz}
\usetikzlibrary{patterns}
\usetikzlibrary{decorations,arrows}
\usetikzlibrary{decorations.pathmorphing}
\usepgflibrary{decorations.pathreplacing} 
\usetikzlibrary{positioning}
\definecolor{verylightgray}{rgb}{0.001,0.001,0.001}

\usepackage{todonotes}

\usepackage{enumitem}
\setdescription{leftmargin=0em,labelindent=0\parindent,labelsep = 1em, itemsep=0.1\baselineskip,topsep=0.1\baselineskip,itemsep =\topsep,  listparindent=\parindent, font=\normalfont\normalsize\itshape}
\newcounter{steplistcount}
\newcounter{steplistcounti}
\newcounter{caselistcount}
\newcounter{caselistcounti}
\renewcommand*\thesteplistcount{\arabic{steplistcount}}
\renewcommand*\thesteplistcounti{\arabic{steplistcounti}}
\renewcommand*\thecaselistcount{\arabic{caselistcount}}
\renewcommand*\thecaselistcounti{\arabic{caselistcounti}}
\newlist{stepList}{description}{2}
\setlist[stepList,1]{%
	before={\setcounter{steplistcount}{0}},% Entsprechenden counter zurück setzen
	leftmargin=0em,labelindent=0em,labelsep = \parindent,itemsep=0.1\baselineskip,topsep=0.1\baselineskip,listparindent=\parindent,style = sameline% Layout einstellungen
	,font=\normalfont\normalsize\itshape Step~\stepcounter{steplistcount}\thesteplistcount:~ % das was hinter font kommt wird in jedem item ergänzt
}
\setlist[stepList,2]{%
	before={\setcounter{steplistcounti}{0}},%
	leftmargin=0em,labelindent=0em,labelsep = \parindent,itemsep=0.1\baselineskip,topsep=0.1\baselineskip,listparindent=\parindent,style = sameline%
	,font=\normalfont\normalsize\itshape Step~ \stepcounter{steplistcounti}\thesteplistcount.\thesteplistcounti:~ % das was hinter font kommt wird in jedem item ergänzt
}

\newlist{caseList}{description}{2}
\setlist[caseList,1]{%
	before={\setcounter{caselistcount}{0}},%
	leftmargin=0em,labelindent=0em,labelsep = \parindent, listparindent=\parindent,%
	style = sameline,font=\normalfont\normalsize\itshape Case~\stepcounter{caselistcount}\thecaselistcount:~%
}

\setlist[caseList,2]{%
	before={\setcounter{caselistcounti}{0}},%
	leftmargin=0em,labelindent=0em,labelsep = \parindent,listparindent=\parindent,itemsep=0.1\baselineskip,topsep=0.4\baselineskip,%%
	style = sameline,font=\normalfont\normalsize\itshape Case~ \stepcounter{caselistcounti}\thecaselistcount.\thecaselistcounti:~ %
}

\newtheoremstyle{thm}% name
	{11pt}       % Space above
	{11pt}      % Space below
	{\itshape}  % Body font
	{}          % Indent amount (empty = no indent, \parindent = para indent)
	{\bfseries} % Thm head font
	{:}         % Punctuation after thm head
	{ }      % Space after thm head: " " = normal interword space;
	{}          % Thm head spec (can be left empty, meaning 'normal')

\theoremstyle{thm}
\newtheorem{lemma}{Lemma}
\newtheorem{theorem}{Theorem}

\newtheorem{corollary}{Corollary}
\theoremstyle{definition}

\theoremstyle{remark}

\newtheorem{remark}{Remark}

\newtheoremstyle{definition}% name
	{11pt}       % Space above
	{0pt}      % Space below
	{\normalfont}  % Body font
	{}          % Indent amount (empty = no indent, \parindent = para indent)
	{\bfseries} % Thm head font
	{:}         % Punctuation after thm head
	{ }      % Space after thm head: " " = normal interword space;
	{}          % Thm head spec (can be left empty, meaning 'normal')
\theoremstyle{definition}

\newcommand{\sett}[2]{\left\{#1\,\middle|\,#2\right\}}
\newcommand{\sset}[2]{\{#1\,|\,#2\}}

\newcommand{\eps}{\varepsilon}
\newcommand{\OPT}{\mathrm{OPT}}
\newcommand{\jobs}{\mathcal{J}}

\newcommand{\schedAlg}{\texttt{Alg}}
\newcommand{\free}{\mathrm{idle}}
\newcommand{\work}{\mathcal{W}}
\newcommand{\startPoints}{\mathcal{S}}

\newcommand{\items}{\mathcal{I}}
\newcommand{\itemsL}{\mathcal{L}}
\newcommand{\itemsH}{\mathcal{H}}
\newcommand{\itemsV}{\mathcal{V}}
\newcommand{\itemsS}{\mathcal{S}}
\newcommand{\itemsM}{\mathcal{M}}

\newcommand{\pT}[1]{p(#1)}
\newcommand{\mN}[1]{q(#1)}

\newcommand{\iW}[1]{w(#1)}
\newcommand{\iH}[1]{h(#1)}
\newcommand{\areaI}[1]{\mathcal{A}(#1)}

\newcommand{\area}[1]{\mathcal{A}(#1)}

\newcommand{\Ts}{T_A}
\DeclareMathOperator{\disjCup}{\dot\cup}

\newcommand{\NN}{\mathbb{N}}

\newcommand{\Oh}{\mathcal{O}}

\newcommand{\NP}{\mathrm{NP}}
\newcommand{\Poly}{\mathrm{P}}

\title{Linear Time Algorithms for Multiple Cluster Scheduling and Multiple Strip Packing \thanks{Research was supported by German Research Foundation (DFG) project JA 612 /20-1}%
}
\author{Klaus Jansen, Malin Rau\\
Institute of Computer Science, University of Kiel, 24118 Kiel, Germany\\
\{kj,mra\}@informatik.uni-kiel.de}
\date{}

\begin{document}
\maketitle

\begin{abstract}
We study the \acl{MCS} problem and the \acl{MSP} problem.
For both problems, there is no algorithm with approximation ratio better than $2$ unless $\Poly = \NP$. 
In this paper, we present an algorithm with approximation ratio $2$ and running time $\Oh(n)$ for both problems. 
While a $2$ approximation was known before, the running time of the algorithm is at least $\Omega(n^{256})$ in the worst case. 
Therefore, an $\Oh(n)$ algorithm is surprising and the best possible.
We archive this result by calling an AEPTAS with approximation guarantee $(1+\eps)\OPT +p_{\max}$ and running time of the form $\Oh(n\log(1/\eps)+ f(1/\eps))$ with a constant $\eps$ to schedule the jobs on a single cluster. 
This schedule is then distributed on the $N$ clusters in $\Oh(n)$.
Moreover, this distribution technique can be applied to any variant of of Multi Cluster Scheduling for which there exists an AEPTAS with additive term $p_{\max}$.
%Furthermore, we give the tools to find a 2-approximation for any variant of Multi Cluster Scheduling with running time $\Oh(n)$ if there exists an AEPTAS with approximation ratio $(1+\eps)\OPT +p_{\max}$ and running time of the form $\Oh(n\log(1/\eps)+ f(1/\eps))$ for the single Cluster version.

While the above result is strong from a theoretical point of view, it might not be very practical due to a large hidden constant caused by calling an AEPTAS with a constant $\eps \geq 1/8$ as subroutine. 
Nevertheless, we point out that the general approach of finding first a schedule on one cluster and then distributing it onto the other clusters might come in handy in practical approaches.
We demonstrate this by presenting a practical algorithm with running time $\Oh(n\log(n))$, with out hidden constants, that is a $9/4$-approximation for one third of all possible instances, i.e, all instances where the number of clusters is dividable by $3$, and has an approximation ratio of at most $2.3$ for all instances with at least $9$ clusters.
%This algorithm uses the partitioning technique as well.
\end{abstract}

\begin{acronym}
	\acro{MCS}[MCS]{Multiple Cluster Scheduling}
	\acro{PTS}[PTS]{Parallel Task Scheduling}
	\acro{MSP}[MSP]{Multiple Strip Packing}
	\acro{SP}[SP]{Strip Packing}
\end{acronym}

\section{Introduction}
In this paper, we study two problems \acl{MCS} and \acl{MSP}.
In the optimization problem \acf{MCS}, we are given $n\in \NN$ parallel jobs $\jobs$ and $N\in \NN$ clusters. 
Each cluster consists of $m\in \NN$ identical machines and each job $j \in \jobs$ has a processing time $\pT{j} \in \NN$ as well as a machine requirement $\mN{j}\in \NN_{\leq m}$.
We define the work of a job $j$ as $\work(j) := \pT{j}\cdot \mN{j}$ and define the work of a set of jobs $\jobs'$ as $\work(\jobs') := \sum_{j \in \jobs'}\work(j)$.
A schedule $S$ of the jobs consists of two functions $\sigma: \jobs \rightarrow \NN$ which assigns jobs to starting points and $\rho: \jobs \rightarrow \{1,\dots N\}$, which assigns jobs to the clusters. 
The objective is to find a feasible schedule of all the jobs, which minimizes the makespan, i.e., which minimizes $\max\{p_j+\sigma(j)|j\in \jobs\}$.
A schedule is feasible if at every time $\tau \in \NN$ and any Cluster $i \in \NN$ the number of used machines is bounded by $m$, i.e., if $\sum_{j \in \jobs, \sigma(j) \leq \tau < \sigma(j)+p_j, \rho(j) = i}q_j \leq m$ for all $i\in \NN$ and $\tau \in \NN$.
If the number of clusters is bounded by one, the problem is called \acf{PTS}.
Note that we can assume that $n> N$ since otherwise an optimal schedule would place each job alone on a personal cluster and thus the problem is not hard.

The other problem that we consider is a closely related variant of \ac{MCS}, called \acf{MSP}. 
%The optimization problem \ac{MSP} is similar to the problem \ac{MCS}. 
The main difference is that the jobs have to be allocated on contiguous machines. 
In the Problem \ac{MSP}, we are given $n \in\NN$ rectangular items $\items$ and $N \in \NN$ strips. 
Each strip has an infinite height and the same width $W \in \NN$.
Each item $i \in \items$ has a width $\iW{i}$ and a height $\iH{i}$.
The objective is to find a feasible packing of the items into the strips such that the packing height is minimized.
A packing is feasible if all the items are placed overlapping free into the strips.
If the number of clusters is bounded by one, the problem is called \acf{SP}.

\acl{SP} and \acl{PTS} are classical optimization problems and the extension of these problems to multiple strips or clusters comes natural.
Furthermore, these problems can be motivated by real world problems. 
One example, as stated in \cite{YeHZ11}, is the following:
In operating systems, \ac{MSP} arises in the computer grid and server consolidation \cite{MartyH07}. 
In the system supporting server consolidation on many-core chip multi processors, multiple server applications are deployed onto virtual machines. 
Every virtual machine is allocated several processors and each application might require a number of processors simultaneously.
Hence, a virtual machine can be regarded as a cluster and server applications can be represented as parallel tasks. 
Similarly, in the distributed virtual machines environment, each physical machine can be regarded as a strip while virtual machines are represented as rectangles. 
It is quite natural to investigate the packing algorithm by minimizing the maximum height of the strips. 
This is related to the problem of maximizing the throughput, which is commonly used in the area of operating systems.

In this paper, we consider approximation algorithms for \ac{MCS} and \acf{MSP}.
We say an approximation algorithm $A$ has an (absolute) approximation ratio $\alpha$, if for each instance $I$ of the problem it holds that $A(I) \leq \alpha\OPT(I)$.
If an algorithm $A$ has an approximation ratio of $\alpha$, we say its result is an $\alpha$-approximation.
A family of algorithms consisting of algorithms with approximation ratio $(1+\eps)$ is called polynomial time approximation scheme (PTAS), and a PTAS whose running time is bounded by a polynomial in both the input length $\mathrm{SIZE}(I)$ and $1/\eps$ is called fully polynomial (FPTAS).
If the running time of a PTAS is bounded by a function of the form $\mathrm{poly}(\mathrm{SIZE}(I)) \cdot f(1/\eps)$, where $f$ is an arbitrary function, we say the running time is \textit{efficient} and call it an efficient PTAS or EPTAS.
An algorithm $A$ has an asymptotic approximation ratio $\alpha$ if there is a constant $c$ such that $A(I) \leq \alpha\OPT(I) + c$
and we denote a polynomial time approximation scheme with respect to the asymptotic approximation ratio as an A(E)PTAS.

Zhuk~\cite{Zhuk06} proved that \ac{MCS} and \ac{MSP} cannot be approximated better than $2$ unless $P = NP$. 
There is an algorithm by Ye, Han and Zhang~\cite{YeHZ11} which finds a $2+\eps$-approximation to the optimal solution for each instance of \ac{MCS} or \ac{MSP}.
This algorithm needs to solve an EPTAS for Scheduling On Identical Machines as a subroutine.
The algorithm with the best running time for this problem is currently given by \cite{JansenRohwedder19} and it is bounded by $2^{\mathcal{O}(1/\varepsilon \log^2(1/\varepsilon))} +\mathrm{poly}(n)$. 
As a result the running time of the algorithm from Ye, Han and Zhang~\cite{YeHZ11} is bounded by $2^{\mathcal{O}(1/\varepsilon \log^2(1/\varepsilon))} +\mathrm{poly}(n)$, using \cite{JansenRohwedder19} and corresponding 2-approximation algorithms for \acl{PTS} , e.g., the List-Scheduling algorithm by Garay and Graham \cite{GareyG75}, and \acl{SP}, e.g., Steinbergs-algorithm~\cite{Steinberg97}.
For \ac{MCS}, the approximation ratio of $(2+\eps)$ was improved by 
Jansen and Trystram~\cite{JansenT16} to an algorithm with approximation ratio of 2 and it has a worst case running time of $\Omega(n^{256})$ since it uses an algorithm with running time $n^{\Omega(1/\eps^{1/\eps})}$ with constant $\eps = 1/4$ as a subroutine. 
Furthermore, for \ac{MSP} there is an algorithm by \cite{BougeretDJOT09} that has a ratio of $2$ as well.
The worst case running time of this algorithm is of the form $\Omega(n^{256})$ as well, for the same reasons.

However, since the worst-case running time for these algorithms with an approximation ratio close to or exactly $2$ is so large, work has been done to improve the runtime at the expense of the approximation ratio. 
There is a faster algorithm by Bougeret et al.~\cite{BougeretDJOT10} which guarantees an approximation ratio of $5/2$ and has a running time of $\mathcal{O}(\log(n p_{\max})n(N+\log(n)))$. 
Note that the Multifit algorithm for Schedulin On Identical Machines has an approximation ratio of $13/11$ and a running time of at most $\Oh(n \log(n) + n\log(N)\log(\areaI{\items}/N))$, see \cite{yue1990exact}. 
Hence using this algorithm as a subroutine in \cite{YeHZ11}, we find a $26/11 \approx 2.364$ approximation. 
In~\cite{BougeretDTJR15} they present an algorithm with approximation ratio $7/3$ with running time $\mathcal{O}(\log(n p_{\max})N(n+\log(n)))$. 
Furthermore, they present a fast algorithm with approximation ratio $2$ and the same running time for the case that the job with the largest machine requirement needs less than $m/2$ machines. 
For \ac{MCS} and \ac{MSP}, we present $2$-approximations, where we managed to improve the running time drastically with regard to the $\mathcal{O}$-notation.

\begin{theorem}
\label{thm:ONAlg}
There is an algorithm for \ac{MCS} with approximation ratio $2$ and running time $\mathcal{O}(n)$ if $N > 2$, and running time $\mathcal{O}(n \log (n))$ if $N \in \{1,2\}$.
\end{theorem}
\begin{theorem}
	\label{thm:MSP}
	There is an algorithm for \ac{MSP} with approximation ratio $2$ and running time $\mathcal{O}(n)$ if $N > 2$, and running time $\mathcal{O}(n\log^2(n)/\log(\log(n)))$ if $N \in \{1,2\}$.
\end{theorem}

Note that the running time of these algorithms is the best possible from a theoretical point of view with respect to the $\mathcal{O}$-notation for $N \geq 3$. 
Since we need to assign a start point to each job, we cannot assume that there is an algorithm for \ac{MCS} with running time strictly faster than $\Omega(n)$.

To achieve these results, we use as a subroutine an AEPTAS for the optimization problem \acf{PTS} and \acf{SP} respectively. 
\ac{PTS}  is similar to the problem \ac{MCS} for the special case that only one cluster is given, while \acf{SP} corresponds to \ac{MSP} where $N = 1$. 
Regarding \ac{PTS}, we improved the running time of an algorithm by Jansen \cite{Jansen12PrallelTasks} and developed an AEPTAS.
For \acf{SP}, we find an AEPTAS as well. 
However, the running time depending on $1/\eps$ is worse than in the AEPTAS for \ac{PTS}.
Note that this algorithm is the first AEPTAS for \ac{SP} that has an additive term of $h_{\max}$.
\begin{theorem}
\label{thm:AEPTASPTS}
There is an algorithm for \ac{PTS} with ratio $(1+\eps)\OPT +p_{\max}$ and running time $\mathcal{O}(n\log(1/\eps)+ \log(n)/\eps^2) + \mathcal{O}_{\eps}(1)$.
\end{theorem}
\begin{theorem}
	\label{thm:SP}
	There is an algorithm for \ac{SP} with ratio $(1+\eps)\OPT +h_{\max}$ and running time $\mathcal{O}(n\log(1/\eps)+ \log(n)/\eps^2) + \mathcal{O}_{\eps}(1)$.
\end{theorem}

This algorithms can be used to find an AEPTAS for \ac{MCS} and \ac{MSP} as well by cutting the solution for one cluster or strip into segments of height $(1+\eps)\OPT$. 
The jobs overlapping the cluster borders add further $p_{\max}$ to the approximation ratio resulting in a additional algorithm for \ac{MCS} with approximation guarantee $(1+\eps)\OPT+p_{\max}$.

\begin{theorem}
\label{thm:AEPTASMCS}
There are algorithms for \ac{MCS} and \ac{MSP} with ratio $(1+\eps)\OPT +p_{\max}$ and running time $\mathcal{O}(n\log(1/\eps)+ \log(n)/\eps^2) + \mathcal{O}_{\eps}(1)$.
\end{theorem}

The algorithm from Theorem \ref{thm:ONAlg} uses the algorithm from Theorem \ref{thm:AEPTASPTS} as a subroutine with a constant value $\eps = 1/8$ if $N = 2$, $\eps = 1/5$ if $N = 5$, and $\eps \in [1/4,1/3]$ otherwise. 
As a result, the running time of the algorithm can be rather large, while the $\mathcal{O}$-notation suggests otherwise since it hides all the constants.  
Due to this fact, we have developed a truly fast algorithm where the most expensive part is sorting the jobs. 
However, this improved running time yields a slight loss in the approximation factor. 

\begin{theorem}
\label{thm:fastApproximationMSCP}
There is a fast $\mathcal{O}(n\log(n))$ algorithm for \ac{MCS} with approximation ratio $9/4$ if $N = 3i$, $(9i+5)/(4i+2)$ if $N = 3i+1$, and $(9i+10)/(4i+4)$ if $N = 3i +2$ for some $i \in \mathbb{N}$.
\end{theorem} 

Note that the approximation ratio of the algorithm from Theorem \ref{thm:fastApproximationMSCP} is worse than $7/3$ for the cases that $N \in \{2,5\}$ and exactly $7/3$ for the case that $N \in \{4,8\}$. 
However if $N \geq 9$, the approximation ratio is bounded by $2.3$, and $(9i+5)/(4i+2))\OPT$ as well as $((9i+10)/(4i+4))\OPT$ converge to $9/4$ for $i\rightarrow \infty$.

\subsection{Related Work}
We repeat and summarize the results for the variant of \ac{MCS} and \ac{MSP} studied in this paper in Table \ref{tab:OverviewOfResults}.

\begin{table}[ht]
	\begin{tabular}{l|p{0.1\textwidth}|p{0.5\textwidth}|l}
		Problem & Ratio & Remarks & Source\\
		\hline
		\ac{MCS}, \ac{MSP} & $2+\eps$ & Needs solving of Scheduling on Identical Machines with ratio $1+\eps/2$ & \cite{YeHZ11}\\
		\hline
		\ac{MCS} & $2$ & Worst case running time at least $\Omega(n^{256})$; can handle clusters with different sizes  & \cite{JansenT16}\\
		\ac{MSP} & $2$ & Worst case running time at least $\Omega(n^{256})$ & \cite{BougeretDJOT09}\\
		\hline
		\ac{MCS}, \ac{MSP} & AFPTAS & Additive constant in $\Oh(1/\eps^2)$, and $\Oh(1)$ for large values for $N$ & \cite{BougeretDJOT09}\\
		\hline
		\hline
		\ac{MCS} & $3$ & Fast algorithm that can handle clusters with different sizes&\cite{SchwiegelshohnTY08}\\
		\hline
		\ac{MCS} &$5/2$& Fast algorithm & \cite{BougeretDJOT10}\\
		\hline
		\ac{MCS} &$7/3$& Fast algorithm & \cite{BougeretDTJR15}\\
		\hline
		\ac{MCS} &$2$& Fast algorithm; requires $\max_{j \in \jobs} \mN{j} \leq \nicefrac{1}{2} \cdot m$ & \cite{BougeretDTJR15}\\
		\hline
		\hline
		\ac{MCS}, \ac{MSP}& $2\OPT$& Running time $\mathcal{O}(n)$ for $N \geq 3$ and $\mathcal{O}(n\log(n))$ for \ac{MCS} and $N=2$, $\mathcal{O}(n \log^2(n))$ for \ac{MSP} and $N = 2$ & This paper\\
		\hline
		\ac{MCS}, \ac{MSP}&  AEPTAS & Additive term $p_{\max}$; linear in $n$ & This paper\\
		\hline
		\ac{MCS} & & Approximation ratio $9/4$ if $N \bmod 3 = 0$ and if $N$ is large& This paper\\
		\hline
		\ac{PTS}, \ac{SP}&AEPTAS& Additive term $p_{\max}$; linear in $n$  & This paper\\
	\end{tabular}	
	\caption{Overview of the results for \ac{MCS} and \ac{MSP}.}
	\label{tab:OverviewOfResults}
\end{table}

\ac{MCS} has also been studied for the case that clusters do not need to have the same number of machines. 
It is still $NP$-hard to approximate this problem better than $2$ \cite{Zhuk06}. 
Furthermore, it was proven in \cite{SchwiegelshohnTY08} and \cite{TchernykhRAKGZ05} that the List Schedule even cannot guarantee a constant approximation ratio for this problem. 
 
The first algorithm was presented by Tchernykh et al. \cite{TchernykhRAKGZ05} and has an approximation ratio of $10$. 
This ratio was improved to a $3$-approximation by Schwiegelshohn et al.~\cite{SchwiegelshohnTY08}, which is given by an online non-clairvoyant algorithm where the processing times are not known beforehand. 
Later, the algorithm was extended by Tchernykh et al.~\cite{TchernykhSYK10} to the case where jobs have release dates changing the approximation ratio to $2e +1$.
Bougeret et al. \cite{BougeretDJOT10Fast} developed an algorithm with approximation ratio $2.5$ for this case.
This algorithm needs the constraint that the largest machine requirement of a job is smaller than the smallest number of machines available in any given cluster.
This ratio was improved by Dutot et al. \cite{DutotJRT13} by presenting an algorithm with approximation ratio $(2+\eps)$. 
The currently best algorithm for this problem matches the lower bound of $2$ \cite{JansenT16}, but has a large running time of $\Omega(n^{256})$.

%\ac{MCS} and \ac{MSP} have also been studied in for the online version of this problem. In this variant, the items of jobs appear one by one over time and the decision where to place the item has to be done immediately without knowledge of the following items.
%There have been studied two variants of this online version.
%In the first model both assignment to the strip or cluster and packing have been done before the next item is seen. 
%If we assign items immediately (online) but pack them later (offline), the model is called \textit{semi-online}.
%The first online algorithm for \ac{MSP} considered the semi-online model, see Zhuk \cite{Zhuk06}.
%It assigns the jobs to a strip in an online manner, but packs items within a strip after all items appear (and thus an offline strip packing algorithm can be used at this stage). 
%Infact the algorithm in \cite{Zhuk06} uses the bottom-left decreasing algorithm \cite{BakerCR80} for the offline version of Strip Packing and showed that the competitive ratio is 10.
%In \cite{YeHZ11} the authors study the (pure) online version of \ac{MSP} and consider several deterministic and randomized algorithms. 
%The maximum competitive ratio is 6.6623 for a single strip, and the competitive ratio decreases when the number of strips increases.
%Recall that the algorithms mentioned in \cite{SchwiegelshohnTY08} and \cite{TchernykhSYK10} are also online algorithms, but they use other underlying models.
%
%\cite{YeM12}
%\cite{YeZ16}

\subsection*{Organization of this Paper}
The $\Oh(n)$ algorithm consists of two steps. 
First, we use an AEPTAS for \ac{MCS} or \ac{MSP} to find a schedule on two clusters, one with makespan at most $(1+\eps)N\OPT$ and the other with mackespan at most $p_{\max}\leq \OPT$.
This schedule on the two clusters is then distributed onto the $N$ clusters using a partitioning technique, as we call it.
This partitioning technique is the main accomplishment of this paper and presented in Section \ref{sec:Partitioning}.
The AEPTAS for \ac{MCS} can be found in Section \ref{sec:scheduling} while the AEPTAS for \acl{MSP} can be found in Section \ref{sec:StripPacking}.
In Section \ref{sec:FastAlgorithm}, we present the algorithm from Theorem \ref{thm:fastApproximationMSCP} that finds an approximation without the need to call the AEPTAS as a subroutine but uses te partitioning technique as well.

\section{Partitioning Technique}
\label{sec:Partitioning}
In this section, we describe the central idea which leads to a linear running time algorithm.  
Indeed this technique can be used for any problem setting  where there is an AEPTAS with approximation ratio $(1+\eps)\OPT+p_{\max}$ for the single cluster version. In this context $p_{\max}$ is the largest occurring size in the minimization dimension, e.g. the maximal processing time or maximal height of the packing. 

Instead of scheduling the jobs on $N$ clusters, we first schedule them on two clusters $C_1$ and $C_2$. 
In a second step, we distribute the scheduled jobs to $N$ clusters.
In the following, let $\OPT$ be the height of an optimal schedule on $N$ clusters for a given instance $I$. 
Since there is a schedule with makespan $\OPT$ on $N$ clusters, there exists a schedule on one cluster with makespan at most $N\cdot\OPT$.  
Assume there is an algorithm \schedAlg\ which schedules the jobs on two clusters $C_1$ and $C_2$ such that the makespan of $C_1$ is at most $(1+\eps)N\cdot \OPT$ and $C_2$ has a makespan of at most $\OPT$. 
The algorithm mentioned in Theorem \ref{thm:AEPTASPTS} is an example of such an algorithm and we will present it in Section \ref{sec:scheduling}.

\begin{lemma}
	\label{lma:partitioningLemma}
	Let an algorithm \schedAlg\  be given that schedules the jobs on two clusters $C_1$ and $C_2$ such that the makespan of $C_1$ is at most $(1+\eps)N\OPT$ and $C_2$ has a makespan of at most $\OPT$ and which has a running time of $\mathcal{O}(n\cdot f(\eps))$. 
	Furthermore, let \texttt{Alg2} be an algorithm that finds for the single cluster variant a schedule or packing with height at most $2 \cdot \max\{\work(\jobs'), p_{\max}\}$ in $\Oh(\texttt{Alg2})$ time for any given set of jobs $\jobs'$.
	
	We can find a schedule on $N\geq 2$ clusters with makespan $2\OPT$ in $\mathcal{O}(n + n\cdot f(\lfloor N/3\rfloor/N)) = \Oh(n)$ operations if $N >2$ and $\mathcal{O}(\texttt{Alg2} + n\cdot f(1/8)) = \Oh(\texttt{Alg2})$ operations if $N = 2$. (Note that $\lfloor N/3\rfloor/N \in [1/5,1/3]$, and hence can be handled as a constant)
\end{lemma}

%We will prove this lemma in the following two subsections.

\subsection*{The case $N > 2$}
In the following, we will describe how to distribute a schedule given by \schedAlg\ to $N$ new clusters, and which value we have to choose for $\epsilon$ in \schedAlg\ to get the desired approximation ratio of $2$.
The partitioning algorithm distinguishes three cases: $N = 3i, N = 3i+1$ and $N = 3i +2$ for some $i \in \mathbb{N}_{\geq 1}$ and chooses the value for $\varepsilon$ dependent on this $N$, such that $\eps \in [1/5,1/3]$. 
%We partition the schedule in all three cases in the same way.
%We just have to handle the last part of the schedule differently. 
In the following, when speaking of a schedule the processing time is on the vertical axis while the machines are displayed on the horizontal axis, see Figure \ref{fig:N1}.

In the following distributing algorithm, we draw horizontal lines at each multiple of $2\Ts$, where $\Ts \leq \OPT$ is a value which depends on the makespan of the schedule defined by \schedAlg\ and will be specified dependent on $N$ in the later paragraphs. 
Let $i \in \mathbb{N}$ and consider the jobs which start at or after $2i\Ts$ and end at or before $2(i+1)\Ts$.
We remove these jobs from $C_1$ and schedule them on a new cluster such that they keep their relative position. 
We say these new clusters have type $A$. 

%The jobs between two consecutive lines $2i\OPT$ and $2(i+1)\OPT$, which are not intersected by these lines, are scheduled on one cluster, having makespan $2\OPT$. 
%We call these clusters type $A$. 
Next, consider the set of jobs cut by the horizontal line at $2i\Ts$.
All these jobs have a processing time of at most $p_{\max} \leq \OPT$ and they can be scheduled at the same time without violating the machine constraint. 
In a new cluster, we can schedule two of these sets of jobs with makespan $2 p_{\max} \leq 2\OPT$, by letting the first set of jobs start at $0$ and the second set start at $p_{\max}$. 
We say, these clusters have type B. 
\begin{caseList}
%\subparagraph*{Case 1: $N = 3i$}
\item[$N = 3i$.]
If $N = 3i$, we choose $\eps := \lfloor N/3\rfloor/N = 1/3$: 
As a result, the schedule on $C_1$ given by \schedAlg\  has a makespan of $T \leq (4/3)N\OPT = 4i\OPT$ and we define $\Ts := T/(4i) \leq \OPT$. 
We partition the given schedule as described above. 
Since it has a height of $4i\Ts$, we get $2i$ clusters of type A, see Figure \ref{fig:N1}. 
There are $ 4i\Ts / (2\Ts) -1 = 2i-1$ lines at multiples of $2\Ts$. 
Hence, we get $\left\lfloor\frac{2i-1}{2}\right\rfloor = i-1$ clusters of type B. 
The jobs intersecting the last line can be scheduled on one new cluster with makespan $\Ts$. 
On this last cluster after the point in time $\Ts$, we schedule the jobs from the Cluster $C_2$.
Remember, the schedule on $C_2$ has a makespan of at most $\OPT$ and, hence, the makespan of this last cluster is bounded by $2\OPT$ as well. 
In total, we have partitioned the schedule into  $2i + i-1 +1 = 3i = N$ clusters each with makespan at most $2\OPT$.

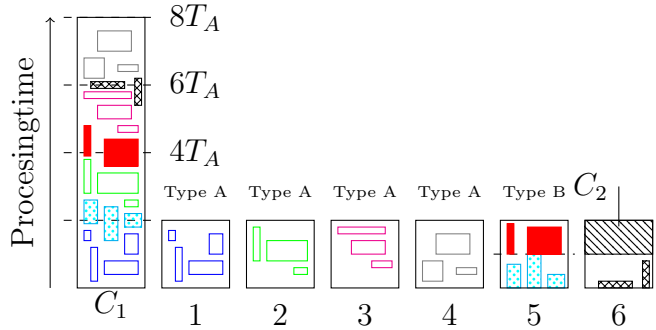
\begin{figure}[ht]
	\centering
	%\resizebox{0.49\textwidth}{!}{%
	\begin{tikzpicture}
	\pgfmathsetmacro{\h}{0.9}
	\pgfmathsetmacro{\w}{0.9}
	
	\draw (0,0) rectangle (\w,4*\h);
	\node at (0.5 *\w,-0.25*\h) {$C_1$};
	\draw[->] (-0.4*\w,0) -- node[midway, rotate=90,above]{Procesingtime} (-0.4*\w,4*\h);
	
	\draw[dashed] (-0.2*\w, \h)  -- (1.2*\w, \h);
	\draw[dashed] (-0.2*\w, 2*\h) -- (1.2*\w, 2*\h) node[right] {$4\Ts$};
	\draw[dashed] (-0.2*\w, 3*\h) -- (1.2*\w, 3*\h) node[right] {$6\Ts$};
	\draw[dashed] (-0.2*\w, 4*\h) -- (1.2*\w, 4*\h) node[right] {$8\Ts$};
	
	% Jobs
	\draw[blue, pattern color=blue] (0.2*\w,0.1*\h) rectangle (0.3*\w,0.6*\h);
	\draw[blue, pattern color=blue] (0.7*\w,0.5*\h) rectangle (0.9*\w,0.8*\h);
	\draw[blue, pattern color=blue] (0.1*\w,0.7*\h) rectangle (0.2*\w,0.85*\h);
	\draw[blue, pattern color=blue] (0.4*\w,0.2*\h) rectangle (0.9*\w,0.4*\h);
	
	\draw[cyan, pattern=crosshatch dots, pattern color=cyan] (0.4*\w, 0.7*\h) rectangle (0.6*\w,1.2*\h);
	\draw[cyan, pattern=crosshatch dots, pattern color=cyan] (0.1*\w, 0.95*\h) rectangle (0.3*\w,1.3*\h);
	\draw[cyan, pattern=crosshatch dots, pattern color=cyan] (0.7*\w, 0.9*\h) rectangle (0.95*\w,1.1*\h);
	
	\begin{scope}[yshift = 1*\h cm]
	\draw[green, pattern color=green] (0.3*\w,0.4*\h) rectangle (0.9*\w,0.7*\h);
	\draw[green, pattern color=green] (0.7*\w,0.2*\h) rectangle (0.9*\w,0.3*\h);
	\draw[green, pattern color=green] (0.1*\w,0.4*\h) rectangle (0.2*\w,0.9*\h);
	\end{scope}
	
	\begin{scope}[yshift = 1*\h cm]
	\draw[red, fill] (0.4*\w, 0.8*\h) rectangle (0.9*\w,1.2*\h);
	\draw[red, fill] (0.1*\w, 0.95*\h) rectangle (0.2*\w,1.4*\h);
	\end{scope}
	
	\begin{scope}[yshift = 2*\h cm]
	\draw[magenta] (0.1*\w,0.8*\h) rectangle (0.8*\w,0.9*\h);
	\draw[magenta] (0.6*\w,0.3*\h) rectangle (0.9*\w,0.4*\h);
	\draw[magenta] (0.3*\w,0.5*\h) rectangle (0.8*\w,0.7*\h);
	\end{scope}
	
	\begin{scope}[yshift = 2*\h cm]
	\draw[black, pattern=crosshatch] (0.85*\w, 0.7*\h) rectangle (0.95*\w,1.1*\h);
	\draw[black, pattern=crosshatch] (0.2*\w, 0.95*\h) rectangle (0.7*\w,1.05*\h);
	\end{scope}
	
	\begin{scope}[yshift = 3*\h cm]
	\draw[gray] (0.1*\w,0.1*\h) rectangle (0.4*\w,0.4*\h);
	\draw[gray] (0.6*\w,0.2*\h) rectangle (0.9*\w,0.3*\h);
	\draw[gray] (0.3*\w,0.5*\h) rectangle (0.8*\w,0.8*\h);
	\end{scope}
	
	% shifted Jobs
	\begin{scope}[xshift = 1.25*\w cm]
	\draw[blue, pattern color=blue] (0.2*\w,0.1*\h) rectangle (0.3*\w,0.6*\h);
	\draw[blue, pattern color=blue] (0.7*\w,0.5*\h) rectangle (0.9*\w,0.8*\h);
	\draw[blue, pattern color=blue] (0.1*\w,0.7*\h) rectangle (0.2*\w,0.85*\h);
	\draw[blue, pattern color=blue] (0.4*\w,0.2*\h) rectangle (0.9*\w,0.4*\h);
	\draw (0*\w,0) rectangle (1*\w,\h);
	\node at (0.5*\w,-0.4*\h) {$1$};
	\node at (0.5*\w,1.4*\h) {\tiny{Type A}};
	\end{scope}
	
	\begin{scope}[xshift = 2.5*\w cm]
	\draw[green, pattern color=green] (0.3*\w,0.4*\h) rectangle (0.9*\w,0.7*\h);
	\draw[green, pattern color=green] (0.7*\w,0.2*\h) rectangle (0.9*\w,0.3*\h);
	\draw[green, pattern color=green] (0.1*\w,0.4*\h) rectangle (0.2*\w,0.9*\h);
	\draw (0*\w,0) rectangle (1*\w,\h);
	\node at (0.5*\w,-0.4*\h) {$2$};
	\node at (0.5*\w,1.4*\h) {\tiny{Type A}};
	\end{scope}
	
	\begin{scope}[xshift = 3.75*\w cm]
	\draw[magenta] (0.1*\w,0.8*\h) rectangle (0.8*\w,0.9*\h);
	\draw[magenta] (0.6*\w,0.3*\h) rectangle (0.9*\w,0.4*\h);
	\draw[magenta] (0.3*\w,0.5*\h) rectangle (0.8*\w,0.7*\h);
	\draw (0*\w,0) rectangle (1*\w,\h);
	\node at (0.5*\w,-0.4*\h) {$3$};
	\node at (0.5*\w,1.4*\h) {\tiny{Type A}};
	\end{scope}
	
	\begin{scope}[xshift = 5*\w cm]
	\draw[gray] (0.1*\w,0.1*\h) rectangle (0.4*\w,0.4*\h);
	\draw[gray] (0.6*\w,0.2*\h) rectangle (0.9*\w,0.3*\h);
	\draw[gray] (0.3*\w,0.5*\h) rectangle (0.8*\w,0.8*\h);
	\draw (0*\w,0) rectangle (1*\w,\h);
	\node at (0.5*\w,-0.4*\h) {$4$};
	\node at (0.5*\w,1.4*\h) {\tiny{Type A}};
	\end{scope}
	
	\begin{scope}[xshift = 6.25*\w cm]
	\draw[cyan, pattern=crosshatch dots, pattern color=cyan] (0.4*\w, 0) rectangle (0.6*\w,0.5*\h);
	\draw[cyan, pattern=crosshatch dots, pattern color=cyan] (0.1*\w, 0) rectangle (0.3*\w,0.35*\h);
	\draw[cyan, pattern=crosshatch dots, pattern color=cyan] (0.7*\w, 0) rectangle (0.95*\w,0.2*\h);
	\draw (0*\w,0) rectangle (1*\w,\h);
	\node at (0.5*\w,-0.4*\h) {$5$};
	\node at (0.5*\w,1.4*\h) {\tiny{Type B}};
	\draw[dashed] (-0.1*\w, 0.5*\h)  -- (1.1*\w, 0.5*\h);
	\end{scope}
	
	\begin{scope}[xshift = 6.25*\w cm]
	\draw[red, fill] (0.4*\w, 0.5*\h) rectangle (0.9*\w,0.9*\h);
	\draw[red, fill] (0.1*\w, 0.5*\h) rectangle (0.2*\w,0.95*\h);
	\end{scope}
	
	\begin{scope}[xshift = 7.5*\h cm]
	\draw[black, pattern=crosshatch] (0.85*\w, 0) rectangle (0.95*\w,0.4*\h);
	\draw[black, pattern=crosshatch] (0.2*\w, 0) rectangle (0.7*\w,0.1*\h);
	\draw[pattern = north west lines](0, 0.5 *\h) rectangle node[midway](A){} (\w,\h);
	\draw (A) -- (0.5*\w,1.5*\h) node[left] {$C_2$};
	\draw (0*\w,0) rectangle (1*\w,\h);
	\node at (0.5*\w,-0.4*\h) {$6$};
	\end{scope}
	
	\end{tikzpicture}
	%}
	\caption{The case $N = 3i$ for $i = 2$.}
	\label{fig:N1}
\end{figure}

%\subparagraph*{Case 2: $N = 3i+1$}
\item[$N = 3i+1$.]
If $N = 3i +1$ for some $i \in \mathbb{N}$, we choose $\eps := \lfloor N/3\rfloor/N = i/(3i +1) \geq 1/4$. 
As a result, the makespan of $C_1$ generated by the algorithm \schedAlg\ is given by $T \leq (1 + i/(3i +1))N\OPT = (4i +1)\OPT$ and we define $\Ts := T /(4i + 1) \leq \OPT$. 
There are $\lceil(4i +1)/2\rceil -1 = 2i$  multiples of $2\Ts$ smaller than $(4i +1)\Ts$, see Figure \ref{fig:N2}.
Above the last multiple of $2\Ts$ smaller than $(4i +1)\Ts$ namely $4i\Ts$, the schedule has a height of at most $\Ts \leq \OPT$ left. 
Hence using the above-described partitioning, we generate $2i$ clusters of type A. 
The jobs intersecting the $2i$ multiples of $2\Ts$ can be placed into $i$ clusters of type B. 
We have left the jobs above $4i\Ts$, which can be scheduled in a new cluster with makespan $\Ts \leq \OPT$. 
Last, we place the jobs from cluster $C_2$ on top of the schedule in the new cluster, such that it has a makespan of at most $\Ts+\OPT \leq 2\OPT$ in total. 
Altogether, we have distributed the given schedule on $2i + i + 1=3i +1 = N$ clusters, such that each of them has a makespan bounded by $2\OPT$.

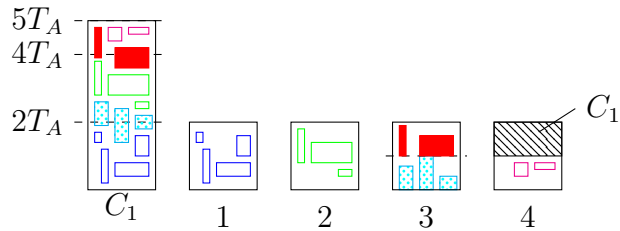
\begin{figure}[ht]
	\centering
	%\resizebox{.9\textwidth}{!}{%
	\begin{tikzpicture}
	\pgfmathsetmacro{\h}{0.9}
	\pgfmathsetmacro{\w}{0.9}
	
	\draw (0,0) rectangle (\w,2.5*\h);
	\node at (0.5 *\w,-0.25*\h) {$C_1$};
	
	\draw[dashed] (-0.2*\w, \h)  node[left] {$2\Ts$} -- (1.2*\w, \h);
	\draw[dashed] (-0.2*\w, 2*\h)  node[left] {$4\Ts$} -- (1.2*\w, 2*\h);
	\draw[dashed] (-0.2*\w, 2.5*\h)  node[left] {$5\Ts$} -- (1.2*\w, 2.5*\h);

	% Jobs
	\draw[blue] (0.2*\w,0.1*\h) rectangle (0.3*\w,0.6*\h);
	\draw[blue] (0.7*\w,0.5*\h) rectangle (0.9*\w,0.8*\h);
	\draw[blue] (0.1*\w,0.7*\h) rectangle (0.2*\w,0.85*\h);
	\draw[blue] (0.4*\w,0.2*\h) rectangle (0.9*\w,0.4*\h);
	
	\draw[cyan, pattern=crosshatch dots, pattern color=cyan] (0.4*\w, 0.7*\h) rectangle (0.6*\w,1.2*\h);
	\draw[cyan, pattern=crosshatch dots, pattern color=cyan] (0.1*\w, 0.95*\h) rectangle (0.3*\w,1.3*\h);
	\draw[cyan, pattern=crosshatch dots, pattern color=cyan] (0.7*\w, 0.9*\h) rectangle (0.95*\w,1.1*\h);
	
	\begin{scope}[yshift = 1*\h cm]
	\draw[green] (0.3*\w,0.4*\h) rectangle (0.9*\w,0.7*\h);
	\draw[green] (0.7*\w,0.2*\h) rectangle (0.9*\w,0.3*\h);
	\draw[green] (0.1*\w,0.4*\h) rectangle (0.2*\w,0.9*\h);
	\end{scope}
	
	\begin{scope}[yshift = 1*\h cm]
	\draw[red, fill] (0.4*\w, 0.8*\h) rectangle (0.9*\w,1.1*\h);
	\draw[red, fill] (0.1*\w, 0.95*\h) rectangle (0.2*\w,1.4*\h);
	\end{scope}
	
	\begin{scope}[yshift = 2*\h cm]
	\draw[magenta] (0.6*\w,0.3*\h) rectangle (0.9*\w,0.4*\h);
	\draw[magenta] (0.3*\w,0.2*\h) rectangle (0.5*\w,0.4*\h);
	\end{scope}
	
	% shifted Jobs
	\begin{scope}[xshift = 1.5*\w cm]
	\draw[blue] (0.2*\w,0.1*\h) rectangle (0.3*\w,0.6*\h);
	\draw[blue] (0.7*\w,0.5*\h) rectangle (0.9*\w,0.8*\h);
	\draw[blue] (0.1*\w,0.7*\h) rectangle (0.2*\w,0.85*\h);
	\draw[blue] (0.4*\w,0.2*\h) rectangle (0.9*\w,0.4*\h);
	\draw (0*\w,0) rectangle (1*\w,\h);
	\node at (0.5*\w,-0.4*\h) {$1$};
	\end{scope}
	
	\begin{scope}[xshift = 3*\w cm]
	\draw[green] (0.3*\w,0.4*\h) rectangle (0.9*\w,0.7*\h);
	\draw[green] (0.7*\w,0.2*\h) rectangle (0.9*\w,0.3*\h);
	\draw[green] (0.1*\w,0.4*\h) rectangle (0.2*\w,0.9*\h);
	\draw (0*\w,0) rectangle (1*\w,\h);
	\node at (0.5*\w,-0.4*\h) {$2$};
	\end{scope}
	
	\begin{scope}[xshift = 4.5*\w cm]
	\draw[cyan, pattern=crosshatch dots, pattern color=cyan] (0.4*\w, 0) rectangle (0.6*\w,0.5*\h);
	\draw[cyan, pattern=crosshatch dots, pattern color=cyan] (0.1*\w, 0) rectangle (0.3*\w,0.35*\h);
	\draw[cyan, pattern=crosshatch dots, pattern color=cyan] (0.7*\w, 0) rectangle (0.95*\w,0.2*\h);
	\draw (0*\w,0) rectangle (1*\w,\h);
	\node at (0.5*\w,-0.4*\h) {$3$};
	\draw[dashed] (-0.1*\w, 0.5*\h)  -- (1.1*\w, 0.5*\h);
	\end{scope}
	
	\begin{scope}[xshift = 4.5*\w cm]
	\draw[red, fill] (0.4*\w, 0.5*\h) rectangle (0.9*\w,0.8*\h);
	\draw[red, fill] (0.1*\w, 0.5*\h) rectangle (0.2*\w,0.95*\h);
	\end{scope}
	
	\begin{scope}[xshift = 6*\h cm]
	\draw[magenta] (0.6*\w,0.3*\h) rectangle (0.9*\w,0.4*\h);
	\draw[magenta] (0.3*\w,0.2*\h) rectangle (0.5*\w,0.4*\h);;
	\draw[pattern = north west lines](0, 0.5 *\h) rectangle node[midway](A){} (\w,\h);
	\draw (A) -- (1.2*\w,1.2*\h) node[right] {$C_1$};
	\draw (0*\w,0) rectangle (1*\w,\h);
	\node at (0.5*\w,-0.4*\h) {$4$};
	\end{scope}
	
	\end{tikzpicture}
	%}
	\caption{The case $N = 3i+1$ for $i = 1$.}
	\label{fig:N2}
\end{figure}

%\subparagraph*{Case 3: $N = 3i+2$}
\item[$N = 3i +2$.]
If $N = 3i +2$, we choose $\eps = \lfloor N/3\rfloor/N = i/(3i+2) \geq 1/5$: As a result, the makespan on $C_1$ generated by \schedAlg\ is bounded by $T \leq (1 + i/(3i +2))N\OPT = (4i +2)\OPT$ and we define $\Ts := T/(4i + 2) \leq \OPT$. 
Thus, there are $(4i +2)/2 -1 = 2i$ vertical lines at the multiples of $2\Ts$, which are strictly larger than $0$ and strictly smaller than $(4i +2)\Ts$, see Figure \ref{fig:N3}. 
As a consequence, we construct $2i+1$ clusters of type A and $i$ clusters of type B. 
The cluster $C_2$ defines one additional cluster of this new schedule. 
In total, we have a schedule on $2i +1 + i + 1 = N$ clusters with makespan bounded by $2\OPT$.
\end{caseList}

\begin{figure}[ht]
	\centering
	%\resizebox{.49\textwidth}{!}{%
	\begin{tikzpicture}
	\pgfmathsetmacro{\h}{0.9}
	\pgfmathsetmacro{\w}{0.9}
	
	\draw (0,0) rectangle (\w,3*\h);
	\node at (0.5 *\w,-0.25*\h) {$C_1$};
	
	\draw[dashed] (-0.2*\w, \h)  -- (1.2*\w, \h);
	\draw[dashed] (-0.2*\w, 2*\h) -- (1.2*\w, 2*\h) node[right] {$4\Ts$};
	\draw[dashed] (-0.2*\w, 3*\h) -- (1.2*\w, 3*\h) node[right] {$6\Ts$};
	
	\draw (7.5*\w,0) rectangle (8.5*\w,\h);
	\node at (8*\w,-0.4*\h) {$5$};
	
	% Jobs
	\draw[blue] (0.2*\w,0.1*\h) rectangle (0.3*\w,0.6*\h);
	\draw[blue] (0.7*\w,0.5*\h) rectangle (0.9*\w,0.8*\h);
	\draw[blue] (0.1*\w,0.7*\h) rectangle (0.2*\w,0.85*\h);
	\draw[blue] (0.4*\w,0.2*\h) rectangle (0.9*\w,0.4*\h);
	
	\draw[cyan, pattern=crosshatch dots, pattern color=cyan] (0.4*\w, 0.7*\h) rectangle (0.6*\w,1.2*\h);
	\draw[cyan, pattern=crosshatch dots, pattern color=cyan] (0.1*\w, 0.95*\h) rectangle (0.3*\w,1.3*\h);
	\draw[cyan, pattern=crosshatch dots, pattern color=cyan] (0.7*\w, 0.9*\h) rectangle (0.95*\w,1.1*\h);
	
	\begin{scope}[yshift = 1*\h cm]
	\draw[green] (0.3*\w,0.4*\h) rectangle (0.9*\w,0.7*\h);
	\draw[green] (0.7*\w,0.2*\h) rectangle (0.9*\w,0.3*\h);
	\draw[green] (0.1*\w,0.4*\h) rectangle (0.2*\w,0.9*\h);
	\end{scope}
	
	\begin{scope}[yshift = 1*\h cm]
	\draw[red,fill] (0.4*\w, 0.8*\h) rectangle (0.9*\w,1.2*\h);
	\draw[red,fill] (0.1*\w, 0.95*\h) rectangle (0.2*\w,1.4*\h);
	\end{scope}
	
	\begin{scope}[yshift = 2*\h cm]
	\draw[magenta] (0.1*\w,0.8*\h) rectangle (0.8*\w,0.9*\h);
	\draw[magenta] (0.6*\w,0.3*\h) rectangle (0.9*\w,0.4*\h);
	\draw[magenta] (0.3*\w,0.5*\h) rectangle (0.8*\w,0.7*\h);
	\end{scope}

	% shifted Jobs
	\begin{scope}[xshift = 1.5*\w cm]
	\draw[blue] (0.2*\w,0.1*\h) rectangle (0.3*\w,0.6*\h);
	\draw[blue] (0.7*\w,0.5*\h) rectangle (0.9*\w,0.8*\h);
	\draw[blue] (0.1*\w,0.7*\h) rectangle (0.2*\w,0.85*\h);
	\draw[blue] (0.4*\w,0.2*\h) rectangle (0.9*\w,0.4*\h);
	\draw (0*\w,0) rectangle (1*\w,\h);
	\node at (0.5*\w,-0.4*\h) {$1$};
	\end{scope}
	
	\begin{scope}[xshift = 3*\w cm]
	\draw[green] (0.3*\w,0.4*\h) rectangle (0.9*\w,0.7*\h);
	\draw[green] (0.7*\w,0.2*\h) rectangle (0.9*\w,0.3*\h);
	\draw[green] (0.1*\w,0.4*\h) rectangle (0.2*\w,0.9*\h);
	\draw (0*\w,0) rectangle (1*\w,\h);
	\node at (0.5*\w,-0.4*\h) {$2$};
	\end{scope}
	
	\begin{scope}[xshift = 4.5*\w cm]
	\draw[magenta] (0.1*\w,0.8*\h) rectangle (0.8*\w,0.9*\h);
	\draw[magenta] (0.6*\w,0.3*\h) rectangle (0.9*\w,0.4*\h);
	\draw[magenta] (0.3*\w,0.5*\h) rectangle (0.8*\w,0.7*\h);
	\draw (0*\w,0) rectangle (1*\w,\h);
	\node at (0.5*\w,-0.4*\h) {$3$};
	\end{scope}
	
	\begin{scope}[xshift = 6*\w cm]
	\draw[cyan, pattern=crosshatch dots, pattern color=cyan] (0.4*\w, 0) rectangle (0.6*\w,0.5*\h);
	\draw[cyan, pattern=crosshatch dots, pattern color=cyan] (0.1*\w, 0) rectangle (0.3*\w,0.35*\h);
	\draw[cyan, pattern=crosshatch dots, pattern color=cyan] (0.7*\w, 0) rectangle (0.95*\w,0.2*\h);
	\draw (0*\w,0) rectangle (1*\w,\h);
	\node at (0.5*\w,-0.4*\h) {$4$};
	\draw[dashed] (-0.1*\w, 0.5*\h)  -- (1.1*\w, 0.5*\h);
	\end{scope}
	
	\begin{scope}[xshift = 6*\w cm]
	\draw[red, fill] (0.4*\w, 0.5*\h) rectangle (0.9*\w,0.9*\h);
	\draw[red, fill] (0.1*\w, 0.5*\h) rectangle (0.2*\w,0.95*\h);
	\end{scope}
	
	\begin{scope}[xshift = 7.5*\h cm]
	\draw[pattern = north west lines](0,0) rectangle node(A){} (\w,0.5*\h);
	\draw (A) -- (1.2*\w,0.7*\h) node[right] {$C_1$};
	\end{scope}
	
	\end{tikzpicture}
	%}
	\caption{The case $N= 3i+2$ for $i = 1$.}
	\label{fig:N3}
\end{figure}
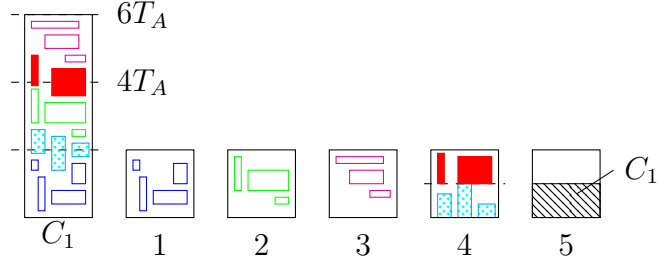

This distribution can be made in $\mathcal{O}(n)$ steps since we have to relocate each job at most once. 
Therefore the algorithm has a running time of at most $\mathcal{O}(n + n\cdot f(\lfloor N/3\rfloor/N)) = \Oh(n)$ since $\lfloor N/3\rfloor/N$ is a constant of size at least $1/5$.

\subsection*{The case $N = 2$}
To find a distribution for this case, we need to make a stronger assumption to the solution of the algorithm \schedAlg. 
Namely, we assume that the second cluster $C_2$ has just $\eps m$ machines. 
As a consequence, the total work of the jobs contained on $C_2$ is bounded by $\eps m\OPT$.

Let us consider the schedule on cluster $C_1$ with makespan $T \leq (1+\eps)2\OPT$. 
In the following, we will assume that $T>2p_{\max}$ since otherwise we have $T \leq 2\OPT$ and we do not need to reorder the schedule any further.
We draw horizontal lines at $\eps T$ and at $T -\eps T$. 
Next, we define two sets of jobs $J_1$ and $J_2$. 
$J_1$ contains all jobs starting before $\eps T$ and $J_2$ contains all jobs ending after $T -\eps T$. 
Note that since $T \leq (1+\eps)2\OPT$, we have that $(1-\eps)T < 2\OPT$. 
Furthermore, $J_1$ and $J_2$ are disjoint if $\eps \leq 1/4$ since $p_{\max} \leq T/2$ and therefore $\eps T +p_{\max} \leq T/4 +T/2 \leq \nicefrac{3}{4}T \leq (1 - \eps)T$. 
Note that the total work of the jobs is bounded by $2\OPT m$ and, hence, $\work(\jobs)/(2 m) \leq \OPT$ .
We distinguish two cases:  

\begin{figure}[ht]
	\centering
	\begin{tikzpicture}
	\pgfmathsetmacro{\h}{1.5}
	\pgfmathsetmacro{\w}{2}
	
	\draw (0,0) rectangle (\w, 2.2*\h);
	\draw[dashed] (-0.25*\w, 1.1*\h) -- (1.25*\w, 1.1*\h) node[right]{$T/2$};
	\draw[dashed] (-0.25*\w, 2*\h) -- (1.25*\w, 2*\h) node[right]{$T -\eps T$};
	\draw[dashed] (-0.25*\w, 2.2*\h) -- (1.25*\w, 2.2*\h) node[right]{$T$};
	\draw[dashed] (-0.25*\w, 0.2*\h) -- (1.25*\w, 0.2*\h) node[right]{$\eps T$};
	\draw[|-|] (0, -0.1*\h) -- node[midway, below]{$m$} (\w, -0.1*\h) ;
	
	\draw[pattern=north east lines] (0.1*\w,0.15*\h) rectangle (0.2*\w, 0.4*\h);
	\draw[pattern=north east lines] (0.4*\w,0.1*\h) rectangle (0.5*\w, 0.6*\h);
	\draw[pattern=north east lines] (0.7*\w,0.05*\h) rectangle (0.85*\w, 0.3*\h);
	
	\draw[pattern=north west lines] (0.15*\w,2.15*\h) rectangle (0.2*\w, 1.4*\h);
	\draw[pattern=north west lines] (0.3*\w,2.1*\h) rectangle (0.5*\w, 1.6*\h);
	\draw[pattern=north west lines] (0.8*\w,2.05*\h) rectangle (0.85*\w, 1.3*\h);
	
	\draw[pattern=north east lines] (0,0) rectangle (\w,0.2*\h);
	\draw[pattern=north west lines] (0,2*\h) rectangle (\w,2.2*\h);
	\end{tikzpicture}
	\caption{The case $N = 2$}
\end{figure}
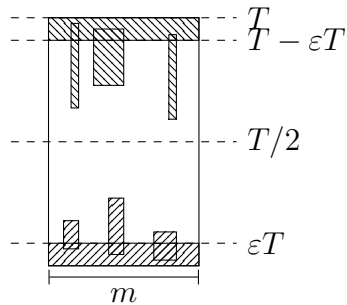

\begin{caseList}
%\subparagraph*{Case 1: $\work(J_1) \leq (1-\eps)\work(\jobs)/2$ or $\work(J_2)\leq (1-\eps)\work(\jobs)/2$}
\item[$\work(J_1) \leq (1-\eps)\work(\jobs)/2$ or $\work(J_2)\leq (1-\eps)\work(\jobs)/2$.] 
Let w.l.o.g $\work(J_2) \leq (1-\eps) \work(\jobs)/2 \leq (1-\eps)m\OPT$. 
We remove all jobs in $J_2$ from the cluster $C_1$. 
As a result this cluster has a makespan of $(1-\eps)T < 2\OPT$. 
The total work of the jobs contained in $C_2$ combined with the jobs in $J_2$ is at most $m\OPT$. 
Therefore, we can use the algorithm \texttt{Alg2} (for example the List-Scheduling algorithm by Garay and Graham \cite{GareyG75}) to find a schedule with makespan at most $2\max\{p_{\max},\work(J_2)\} \leq 2\OPT$. 
Hence, we can find a schedule on two clusters in at most $\mathcal{O}(\texttt{Alg2} + n \cdot f(\eps))$ for this case.

%\subparagraph*{Case 2: $\work(J_1) > (1-\eps)\work(\jobs)/2$ and $\work(J_2)> (1-\eps)\work(\jobs)/2$}
\item[$\work(J_1) > (1-\eps)\work(\jobs)/2$ and $\work(J_2)> (1-\eps)\work(\jobs)/2$.]
Consider the set of jobs $J_3$ scheduled on $C_1$ but not contained in $J_1$ or $J_2$. 
Since the total work of the jobs is at most $\work(\jobs) \leq m \OPT$ it holds that $\work(J_3) \leq \work(\jobs) - \work(J_1) - \work(J_2) = \eps \work(\jobs)\leq 2\eps m\OPT$. 
Let $J(C_1)$ be the set of jobs scheduled on $C_1$ and $J(C_2)$ be the set of jobs scheduled on $C_2$.
We define $J_4 := \{j \in J(C_1)|\sigma(j)+p(j) \leq \eps T\}$ and $J_5 := \{j \in J(C_1)|\sigma(j) \geq (1-\eps)T\}$. 
Clearly, both sets have a total work of at most $\eps mT \leq 2(\eps+\eps^2)m\OPT$ and 
therefore $\work(J_3 \cup J_4 \cup J_5 \cup J(C_2)) \leq (7\eps +4\eps^2)m \OPT$. 
If $\eps = \frac{1}{8}$, these jobs have a total work of at most $m\OPT$ and are scheduled with the algorithm \texttt{Alg2} to find a schedule on one cluster with makespan at most $2\max\{p_{\max},\work(J_3 \cup J_4 \cup J_5 \cup J(C_2))\} \leq 2\OPT$. 

To this point, we have scheduled all jobs except the ones cut by the line $\eps T$ and the jobs cut by the line $(1-\eps)T$. 
We schedule them in the second cluster by starting all the jobs cut by the first line at start point $0$ and the second set of jobs at the start point $p_{\max} \leq \OPT$. 
Note that the partition into the sets $J_1, \dots, J_5$ can be done in $\Oh(n)$ and hence the partitioning step is dominated by the running time of the algorithm \texttt{Alg2}.
%Again the running time is dominated by Tureks algorithm and therefore this repacking needs at most $\mathcal{O}(n \log n)$ operations.

In both cases for $N=2$, we choose $\eps = 1/8$ and, hence, can bound the running time of the algorithm by $\Oh(\texttt{Alg2} + n \cdot f(1/8)) = \Oh(\texttt{Alg2})$. 
\end{caseList}

This concludes the proof of Lemma \ref{lma:partitioningLemma}. 
However, to prove Theorem \ref{thm:ONAlg}, we need to prove the existence of the algorithm \schedAlg, which finds the schedule on the clusters $C_1$ and $C_2$. 
In the next section, we will see one example of such an algorithm. 

As \texttt{Alg2}, we can choose Steinbergs-Algorithm \cite{Steinberg97} in the case of \ac{SP}. 
It has a running time that is bounded by $\Oh(n\log^2(n))$. 
On the other hand for \ac{PTS}, we can use the algorithm by Garay and Graham \cite{GareyG75}, which was optimized by Turek et al. \cite{TurekWY92} to have a running time of $\Oh(n\log(n))$.

A direct conclusion of the lemma is the following corollary.
\begin{corollary}
	For all $N \geq 3 $, given a schedule on two clusters $C_1$ and $C_2$ such that the makespan of $C_1$ is at most $(1+\lfloor N/3\rfloor/N)N\OPT$ and $C_2$ has a makespan of at most $\OPT$, we can find a schedule on $N$ clusters with makespan at most $2\OPT$ in at most $\Oh(n)$ additional steps.
\end{corollary}

Instead of using the algorithm in the next section, first, we can try to use any heuristic or other (fast) approximation algorithm. 
More precisely, we can do the following:
Given a schedule by any heuristic, we remove all the jobs that end after the point in time at which the last job is started and place them on the cluster $C_2$, by starting them all at the same time.
The schedule on $C_2$ obviously has a makespan bounded by $p_{\max} \leq \OPT$.
Next, we check weather the residual schedule on $C_1$ has a makespan of at most $(1+\lfloor N/3\rfloor/N)N\OPT$. 
For example, this can be done by comparing the makespan $T$ on $C_1$ to the lower bound on the optimal makespan $\max\{ p_{\max},\work(\jobs)/m,\pT{\jobs_{>m/2}}\}$, where $\jobs_{>m/2}$ is the set of all jobs with machine requirement larger than $m/2$.
If the makespan $T$ is small enough, i.e., if $T \leq (1+\lfloor N/3\rfloor/N)\max\{ p_{\max},\work(\jobs)/m,\pT{\jobs_{>m/2}}\}$, we will find a $2$-approximation by using the partitioning technique from above. 
Otherwise, we need to use the algorithm from the next section.

\section{An AEPTAS for \acl{PTS}}
\label{sec:scheduling}
In this section, we will present an $AEPTAS$ for \acf{PTS} with an approximation ratio $(1+\eps)\OPT +p_{\max}$ and running time $\mathcal{O}(n\log(1/\eps)) + \mathcal{O}_\eps(1)$. 
We can use this algorithm to find a schedule on the two clusters $C_1$ and $C_2$ needed for the algorithm in Section \ref{sec:Partitioning}.
It is inspired by the algorithm in \cite{Jansen12PrallelTasks} but contains some improvements. 
%After that we will look at some occurring parameters and see how to improve them. 
Furthermore, note the fact that in the following algorithm the processing times of the jobs do not have to be integral. 
Instead, we will discretize them by rounding.

The algorithm works roughly in the following way. 
The set of jobs is partitioned into large, medium, and small jobs, depending on their processing times. 
The medium jobs have a small total work and therefore can be scheduled at the end of the schedule using a 3-approximation algorithm without doing too much harm. 
The large jobs are partitioned into two sets: wide jobs and narrow jobs depending on their machine requirement. 
There are few large wide jobs which makes it possible to guess their starting times. 
The narrow jobs are placed with a linear program for which we guess the number of required machines for each occurring processing time at each possible start point of the schedule.    
After solving this linear program, a few jobs are scheduled fractionally. 
These jobs have a total number of required machines of at most $\gamma m$ for any chosen value $\gamma \in (0,1]$. 
Notice that the choice of $\gamma$ will affect the running time. 
We place these jobs on top of the schedule to gain a $(1+\eps)\OPT + p_{\max}$ approximation, or into an extra cluster to find a solution needed for the algorithm in Section \ref{sec:Partitioning}. 
The small jobs are scheduled with a linear program. 
An overview of the algorithm can be found in Section \ref{subsec:summary}.

We will now present a more detailed approach. 
We use an improved rounding strategy for large jobs compared to~\cite{Jansen12PrallelTasks}, which enables us to improve the running time.
Further, we present a different linear programming approach to schedule the narrow tall jobs. 

%
%The presented algorithm is a ... algorithm, which computes given a makespan $T$ a schedule with makespan $(1+\eps)T + p_{max}$ or decides that there is no schedule with makespan at most $T$. 
%Using this algorithm, we can construct an $(1+\eps)(1+\eps)\OPT + p_{\max}$ approximation algorithm by first computing a lower and upper bound for the makespan and then using a binary search framework to find a $(1+\eps)$-approximation for the smallest possible value for $T$ by enlarging the total running time by a factor of $\log(1/\eps)$.    

\subsection{Simplify}
Let an instance $I = (\jobs, m)$ be given. 
Note that the value $\max\{p_{\max},W(\jobs)/m\}$ is a lower bound on the makespan of the schedule. 
On the other hand, we know by Turek et al. \cite{TurekWY92} that $T := \max\{p_{\max},W(\jobs)/2\}$ is an upper bound on the optimal makespan. 
We can find $T$ in $\mathcal{O}(n)$.

Let $\delta$ and $\mu$ be values dependent on $\eps$.
We partition the set of jobs $\jobs$ into small $\jobs_S := \{j \in \jobs| p_j \leq \mu T\}$, medium $\jobs_M := \{j \in \jobs| \mu T <p_j < \delta T\}$, and large ones $\jobs_L := \{j \in \jobs| \delta T \leq p_j\}$. 
Consider the sequence $\sigma_{0} = \eps$, $\sigma_{i+1}= \sigma_{i}\eps^3$. 
By the pigeonhole principle there exists an $i \in \{0,\dots,1/\eps-1\}$ such that $W(\jobs_M) \leq \eps m\OPT$, when defining $\delta := \sigma_i$ and $\mu := \sigma_{i+1}$. 
We can find these values for $\delta$ and $\mu$ in $\mathcal{O}(n+1/\eps)$. 
Note that $\mu = \eps^3\delta \geq \eps^{3/\eps+3}$.

Resulting in a loss of at most $\eps T$ in the approximation ratio, we can assume that the smallest processing time is at least $\eps T/n$ since adding $\eps T/n$ to each processing time adds at most $n \cdot \eps T/n = \eps T$ to the total makespan.
Therefore, the largest $l$ such that $p_j \in \{\eps^l T,\eps^{l-1}T\}$ is bounded by $\mathcal{O}(\log(n))$ and we know $\delta \geq \min\{\eps/n,\eps^{3/\eps}\}$.
We round the sizes of the jobs by using the following lemma. 
%Again we have to work with $T$ instead of $\OPT$. 
%As a result we know that we can find a schedule with makespan at most $(1+2\eps)T \leq (1+2\eps)(1+\eps)\OPT$ after the rounding.

\begin{lemma}[See \cite{JansenR16}]
	\label{lma:roundingProcessingTimes}
	At a loss of a factor of at most $(1+2\eps)$ in the approximation ratio, we can ensure that each job $j \in \jobs$ with $\eps^{l}  T < p_j \leq \eps^{l-1} T$ for some $l\in\mathbb{N}$ has processing time $p_j = k_j \eps^{l+1}  T$ for $k_j = \lceil p_j /(\eps^{l+1} T)\rceil \in \{1/\eps+1,\dots,1/\eps^2\}$ and a starting time, which is a multiple of $\eps^{l+1} T$ as well. 
\end{lemma}

This rounding can be done in $\mathcal{O}(n)$.
Afterward, there are at most $1/\eps^2$ different processing times between $\eps^{l} T$ and $\eps^{l-1} T$ for each $l \in \{1,\dots, 3/\eps+3\}$. 
Therefore, the number of different processing times of large jobs is bounded by $1/\eps^2 \cdot 3/\eps = 3/\eps^3$ 
since $\delta \geq \eps^{3/\eps}$.
Further, the number of different processing times for medium jobs is bounded by $3/\eps^2$ since the medium jobs have processing times in $(\mu = \eps^3\delta,\delta)$.
Note that the number of different processing times of small jobs is bounded by $\mathcal{O}(\min\{\log(n)/\eps^2,n/\eps\})$ 
since the smallest job has processing time $\eps T /n$.
%$|P_L| \leq \log_{\eps}(\delta)/\eps^2 = ac/\eps^3 +\log_{\eps}(d)/\eps^2\in \mathcal{O}(1/\eps^3)$. 
Additionally, there are at most $\nicefrac{1}{\eps\delta}$ possible starting points for the large jobs. 
We denote the set of starting points for large jobs as $\startPoints$ and the set of their processing times as $P_L$.
After this step, we will only consider the rounded processing times and will denote them as $p_j$ for each job $j \in \jobs$. 
%We call the processing time between two consecutive starting positions in $S$ a time slot.

\subsection{Large Jobs}
\label{sec:LargeJobs}
Let $\gamma m \leq m$ be the width of the second cluster $C_2$ and let $\alpha$ be a constant dependent on $\eps$ and $\gamma$, which we will specify later on. 
We say a job $j \in \jobs_L$ is wide if it uses at least $\alpha m$ machines, and we denote the set of large wide jobs by $\jobs_{L,W}$.
Note that large wide jobs have a processing time larger than $\delta T$ and need at least $\alpha m$ machines while the total work of all jobs in $\jobs$ is bounded by $m T$.
Hence, the total number of them is bounded by $1/(\delta\alpha)$.
Therefore, there are at most $S^{\nicefrac{1}{(\delta\alpha)}}$ possibilities to schedule the jobs in $\jobs_{L,W}$. 
In the algorithm, we will try each of these options.

In the next step, we deal with the large narrow jobs $\jobs_{L,N} := \jobs_L \setminus \jobs_{L,W}$. 
Consider an optimal schedule $S = (\sigma,\rho)$, where we have rounded the processing times of the jobs as described in Lemma \ref{lma:roundingProcessingTimes}. 
For the schedule $S$ and each starting time $s \in \startPoints$, let $m_s$ be the number of machines used by jobs in $\jobs_{L,N}$ that are processed (not only started) at that time, i.e., we define $m_s := \sum_{j \in \jobs_{L,N,s}} q_j$ where $\jobs_{L,N,s}$ is the set of jobs $j \in \jobs_{L,N}$, which have both a start point $s_j \leq s$ and an endpoint $e_j := s_j+p_j > s$. 
Note that jobs ending at $s$, i.e., jobs with $e_j = s_j+p_j = s$, are not part of the set $\jobs_{L,N,s}$.

For each processing time $p \in P_L$ let $q(p)$ be the total number of machines used by jobs with this processing time, i.e $q(p) := \sum_{j \in J_{L,N}, p_j = p}q_j$. Consider the following linear program $LP_{large}$:

\begin{align}
\sum_{p \in P_L}\sum_{\substack{s \in \startPoints(p)_{\leq s'},\\ s + p > s'}} x_{s,p} & = m_s & \forall s' \in \startPoints \label{eq:lpLargeMachineConstraint}\\
\sum_{s \in \startPoints(p)} x_{s,p} & = q(p) & \forall p \in P_L \label{eq:lpLargeJobConstraint}\\
x_{s,p} & \geq 0 & \forall p \in P_L, s \in \startPoints(p)
\end{align}

The variable $x_{s,p}$ defines for each start point $s \in \startPoints$ and each processing time $p \in P$ how many machines are used by jobs with processing time $p$ starting at $s$.
The first inequality ensures that the number of machines required by jobs scheduled at a start point $s$, i.e., jobs from the set $\jobs_{L,N,s}$, equals the number of used machines in the considered optimal schedule. 
The second inequality ensures that for each processing time all the jobs are scheduled. 
Given the considered optimal solution, we generate a solution to this linear program by counting for each starting time $s \in \startPoints$ and each processing time $p \in P_L$ how many machines are used by jobs with processing time $p$ starting at $s$.
%Hence we know that if we choose all the $m_s$ correctly, the linear program has a feasible solution.
This linear program has $|\startPoints|+|P_L|$ conditions and $|\startPoints||P_L|$ variables.  
Since we have $|\startPoints|+|P_L|$ conditions, there are at most $|\startPoints|+|P_L|$ non zero components in a basic solution and for each $p \in P_L$ there has to be at least one non zero component. 

In the algorithm, we guess, (i.e., we try out all the possibilities) which variables are non zero variables in the basic solution. 
There are at most $\mathcal{O}((|\startPoints||P_L|)^{|\startPoints|+|P_L|})$ options. 
We cannot guess the exact values of the variables $x_{s,p}$ in polynomial time. Instead, we guess for each non zero variable $x_{s,p}$ the smallest multiple of $\alpha m$ that is larger than the value of $x_{s,p}$ in the basic solution. 
This can be done in $\mathcal{O}(1/\alpha^{|\startPoints|+|P_L|})$. 
So to find a schedule for the large jobs $\jobs_L$, we use at most $\mathcal{O}(|\startPoints|^{|\jobs_{L,W}|}\cdot (|\startPoints||P_L|/\alpha)^{|\startPoints|+|P_L|})$
% = \mathcal{O}((1/\eps\delta)^{1/\delta\alpha}\cdot (3/\eps^4\delta\alpha)^{1/\eps\delta +3/\eps^3})$ 
guesses in total.

Note that this optimistic guessing, i.e., using the rounded up values for $m_s$, on the one hand ensures that all the narrow large jobs can be scheduled but on the other hand can cause violations to the machine constraints. 
To prevent this machine violation, the algorithm test for each guess whether the job condition (\ref{eq:lpLargeJobConstraint}) is fulfilled for each processing time. 
If this is the case, each value of a non-zero component is reduced by $\alpha m$. 
For these down-sized values, the algorithm test the machine constraint (\ref{eq:lpLargeMachineConstraint}) for each starting point $s \in S$.  
Note that the validation whether the constraints are fulfilled is possible in 
$\Oh((|P_L|+|\startPoints|)^2)$ since for each of the $(|P_L|+|\startPoints|)$ constraints, we have to add at most $(|P_L|+|\startPoints|)$ values for each constraint.
If both conditions are fulfilled, the algorithm tries to schedule the small jobs, see Subsection \ref{sec:smallJobs}.
If the small jobs can be scheduled the guess was feasible.

The actual narrow large jobs from the set $\jobs_{L,N}$ are scheduled only once in the final phase of the algorithm.
When scheduling the jobs in $\jobs_{L,N}$, we use the reduced guessed values. 
We greedily fill the jobs into the guessed starting positions $x_{s,p}$, while slicing jobs vertical if they do not fit totally at that starting position (i.e., if the total number of machines required by jobs with processing time $p$ starting at $s$ is larger than $x_{s,p}$ when adding the machine requirement of the currently considered job) and placing the rest of the job at the next starting position for the processing time $p$. 
We schedule the jobs which cannot be placed at the starting points defined by the values of $x_{s,p}$ (because we reduced these values) either on top of the schedule or on the second cluster $C_2$ depending on what is wanted: the algorithm described in Theorem \ref{thm:AEPTASPTS} or the algorithm needed for Lemma \ref{lma:partitioningLemma}. 
The total width of these jobs shifted to the end of the schedule or to Cluster $C_2$ is bounded by $\alpha m \cdot (|S|+|P_L|)$ since there are at most $|S|+|P_L|$ non zero components and before the reduction by $\alpha m$ all the jobs could be scheduled because the job constraint (\ref{eq:lpLargeJobConstraint}) was fulfilled.
%Since we guessed the next larger multiple of $\alpha m$ for the variables, we have to remove jobs with total machine requirement of at most $\alpha m$ for each non zero variable to make sure that we do not use to many machines in one point of time. 

In the described placement of the narrow large jobs, we have introduced at most one fractional job for each non zero variable and it has a width of at most $\alpha m$.
We remove all these fractional jobs and place them next to the jobs which did not fit.
The machine requirement of the removed fractional jobs can be bounded by $(|S|+|P_L|)\alpha m = (1/2\eps\delta +3/\eps^3) \alpha m$. 
Hence, if $\alpha \leq \eps\delta/4$, we have $2(|S|+|P_L|)\alpha m \leq m$, and we can schedule all the removed jobs (non-fitting ones and fractional ones) at the same time at the end of the schedule without violating the machine constraint, adding at most $p_{\max}$ to the makespan.
On the other hand, if $\alpha \leq \gamma\eps\delta/4$, it follows that $2(|S|+|P_L|)\alpha m \leq \gamma m$, and we can schedule all the removed jobs inside the extra cluster with makespan at most $p_{\max}$ and machine requirement at most $\gamma m$.
In the algorithm, we choose $\alpha$ as needed for the corresponding application.
We need at most $\mathcal{O}(n + |\startPoints|+|P_L|)$ operations to place the narrow large jobs.

%TODO: verbesserungsmöglichkeit: $\alpha = \gamma \eps^{11}/6$, da $p_{\max} \geq \eps^6T$ (sonst $\delta = \eps^3$).

\subsection{Small Jobs}
\label{sec:smallJobs}
We define a layer as the horizontal strip between two consecutive starting points in $\startPoints$ and say layer $l$ is the layer between $l\eps\delta T$ and $(l+1)\eps\delta T$. 
%Consider one layer. 
Note that during the processing time of a layer $l$ the machine requirement of large jobs will not change since large jobs start and end at multiples of $\eps\delta T$.
Let $m_l$ be the number of machines left for small jobs in layer $l$.
Note that this number is fixed by the guesses for the large jobs.

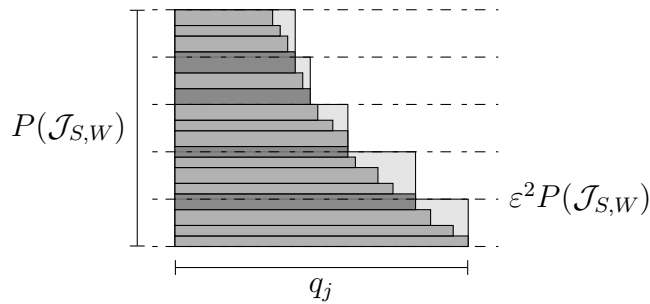
\begin{figure}[ht]
	\centering
	\begin{tikzpicture}
	\pgfmathsetmacro{\h}{0.7}
	\pgfmathsetmacro{\w}{1}
	
	\foreach \i\x\y\xx\yy in {
		0.2/0/0.0/3.9/0.9,
		0.2/0/0.9/3.2/1.8,
		0.2/0/1.8/2.3/2.7,
		0.2/0/2.7/1.8/3.6,
		0.2/0/3.6/1.6/4.5,
		0.5/0/0.0/3.9/0.2,
		0.5/0/0.2/3.7/0.4,
		0.5/0/0.4/3.4/0.7,
		0.9/0/0.7/3.2/1.0,
		0.5/0/1.0/2.9/1.2,
		0.5/0/1.2/2.7/1.5,
		0.5/0/1.5/2.4/1.7,
		0.9/0/1.7/2.3/1.9,
		0.5/0/1.9/2.3/2.2,
		0.5/0/2.2/2.1/2.4,
		0.5/0/2.4/1.9/2.7,
		0.9/0/2.7/1.8/3.0,
		0.5/0/3.0/1.7/3.3,
		0.9/0/3.3/1.6/3.7,
		0.5/0/3.7/1.5/4.0,
		0.5/0/4.0/1.4/4.2,
		0.5/0/4.2/1.3/4.5
	}{
		\draw[fill=gray,fill opacity=\i] (\x*\w,\h*\y) rectangle (\xx*\w,\yy*\h);
	}
	
	\foreach \i in {0,2,3,4,5}
	{ 
		\draw  [dash pattern=on 2pt off 3pt on 4pt off 4pt] (-0.3*\w,\i*0.9*\h) -- (4.3*\w,0.9*\i*\h);	
	}	
	
	\draw  [dash pattern=on 2pt off 3pt on 4pt off 4pt] (-0.3*\w,0.9*\h) -- (4.3*\w,0.9*\h) node[right] {$\eps^2P(\jobs_{S,W})$};
	
	\draw[|-|] (-0.5*\w,0*\h) -- node[midway, left]{$P(\jobs_{S,W})$} (-0.5*\w, 4.5*\h);
	\draw[|-|] (-0.0*\w,-0.4*\h) -- node[midway, below]{$q_j$} (3.9*\w, -0.4*\h);
	\end{tikzpicture} 
	\caption{Linear grouping for small wide jobs}
	\label{fig:LinearGrouping}
\end{figure}

We will partition the small jobs into wide and narrow jobs. 
A small job is wide if it requires at least $\eps m$ machines and narrow otherwise. 
Let $\jobs_{S,W}$ be the set of small wide jobs and $\jobs_{S,N}$ be the set of small narrow jobs.
We will round the machine requirements of the wide jobs using linear grouping, which was first introduced by Fernandez de la Vega \cite{VegaL81}.
The idea of this technique is to sort all the wide jobs by size and stack them on top of each other, such that the widest job is at the bottom and the narrowest job is at the top, see Figure \ref{fig:LinearGrouping}. 
Let $P(\jobs_{S,W})$ be the total processing time of all the small wide jobs. 
We will round the machine requirements of the wide jobs to $1/\eps^2$ sizes. 
For this purpose consider the multiples of $\eps^2P(\jobs_{S,W})$.
We draw a horizontal line at each of these multiples of $\eps^2 P(\jobs_{S,W})$ and define for each job intersected by one of these lines two new jobs, by cutting this job at that line in two parts (for the analysis and description of the rounding; in the algorithm no job will be cut). 
The jobs between two consecutive lines at $i\eps^2P(\jobs_{S,W})$ and $(i+1)\eps^2P(\jobs_{S,W})$ are called \textit{jobs of group} $i$.
For each group $i$, we generate one new job that has processing time $\eps^2 P(\jobs_{S,W})$ and the machine requirement of the widest job in this group. 
We call this job a \textit{size defining job} for the group.
Let $\bar{\jobs}_{S,W}$ be the set of rounded small wide jobs. 

When we round the wide jobs as described, we need $\Omega(n\log(n))$ operations, since we sort the jobs. 
However, we do not need to sort the jobs since we are just interested in the size defining job of each group.

\begin{lemma}
\label{lma:fastRounding}
We can generate the rounded jobs $\bar{\jobs}_{S,W}$ in $\mathcal{O}(n\log(1/\eps))$ operations.
\end{lemma}
\begin{proof}
Define for a given job $j$ the set $\jobs_{S,W,>q_j} := \{i \in \jobs_{S,W}| q_i > q_j\}$ and analogously $\jobs_{S,W,\geq q_j} := \{i \in \jobs_{S,W}| q_i \geq q_j\}$.
For each group, we find the size defining job, by using a modified median algorithm with running time $\mathcal{O}(n)$.
Instead of searching for the $i$th largest job, we search for a job $j$ with $P(\jobs_{S,W,>q_j}) \leq i \eps^2 P(\jobs_{S,W})$ and $P(\jobs_{S,W,>q_j}) > i \eps^2 P(\jobs_{S,W})$ for each $i$ in $\{1, \dots, 1/\eps-1\}$.
Simply using this modified median algorithm for each group leads to $\mathcal{O}(n/\eps^2)$ operations.
 
However, we improve this approach. 
First, we search for the job with the median machine requirement in $\mathcal{O}(|\jobs_{S,W}|)$. 
Afterward, we search for the group size of the group containing this job in $\mathcal{O}(|\jobs_{S,W}|/2)$ and the group size above this group (if existing) in $\mathcal{O}(|\jobs_{S,W}|/2)$ as well. 
The set of jobs, where we do not know the rounded sizes, is now partitioned into two sets containing at most $|\jobs_{S,W}|/2$ jobs each. 
We iterate the process on both sets separately until each group size is found. 

Since there are at most $\mathcal{O}(1/\eps^2)$ groups,
this search can be done in $\mathcal{O}(n\log(1/\eps))$ operations. 
To see this, we consider the following recurrence equation
\begin{align*}
T(n,1) &\leq c n\\
T(n,d) &\leq T(n/2,d') + T(n/2,d'') + cn, \mathrm{for \ }d'+d'' < d
\end{align*}
where $n$ denotes the number of jobs and $d$ denotes the number of values we search for and $c \in \mathbb{N}$. 
To find the job with the median machine requirement and the group sizes of the group containing this item and the group above we need $\mathcal{O}(n)$ operations and hence there is a $c \in \mathbb{N}$ with these properties. 
After the set of jobs is partitioned into two sets such that each set contains at most $n/2$ jobs. 
The total number of sizes we search for is reduced by at least one since in this step we find one or two of them. 
However, the values we search for do not have to be distributed evenly to the sets.
Therefore, this recurrence equation represents the running time of the described algorithm adequately.

We claim that $T(n,d) \leq cn(\log(d)+1)$. 
We have $T(n,1) = cn= cn(\log(1)+1)$ and hence the claim is true for $d = 1$.
For $d\in \mathbb{N}_{\geq 2}$ it follows that 
\begin{align*}
T(n,d) &\leq T(n/2,d') + T(n/2,d'') +cn \\
&\leq c(n/2)(\log(d')+1) +  c(n/2)(\log(d'')+1) +cn  \\
&\leq c(n/2)\cdot (\log(d')+1 + \log(d'')+1) + cn \\
& = c(n/2)(\log(d'd'')) +2cn\\ 
&\leq c(n/2)(\log((d/2)^2)) +2cn \\
%& = c(n/2)2(\log((d/2))) +2cn \\
%& = cn(\log((d/2))) +2cn \\
%& = cn(\log(d) -1) +2cn \\
& = cn\log(d) + cn
\end{align*}
Since in our case we have $d = \eps^2$, this concludes the proof.
\end{proof}

\begin{remark}
	If we schedule the rounded jobs $\bar{\jobs}_{S,W}$ fractionally instead of the original jobs $\jobs_{S,W}$, we need to add at most $\eps T$ to the makespan of the schedule. 
\end{remark}
\begin{proof}
	Consider an optimal schedule of the original small jobs. 
	We can schedule the new jobs fractionally, by replacing all jobs contained in group $i$ by the new job generated for the jobs in the group $(i+1)$. 
	The widest rounded job cannot be scheduled instead of the original jobs, because the machine requirement might be too large. 
	We schedule this job at the end of the schedule. 
	This job has a processing time of $\eps^2 P(\jobs_{S,W})$.
	We know that $P(\jobs_{S,W})\cdot \eps m \leq m T$ since $\work(\jobs) \leq m T$ and each wide job needs at least $\eps m$ machines to be scheduled.
	Hence, it holds that $\eps^2 P(\jobs_{S,W}) \leq \eps T$.
\end{proof}

We say a configuration of wide jobs is a multiset of wide jobs $C :=\{a_{j,C}: j |j \in \bar{\jobs}_{S,W}\}$. 
We say a configuration $C$ requires at most $q$ machines, if $\sum_{j \in \bar{\jobs}_{S,W}}a_{j,C}q_j \leq q$ and define $q(C) := \sum_{j \in \bar{\jobs}_{S,W}}a_{j,C}q_j$. Let $\mathcal{C}_q$ be the set of configurations with machine requirement at most $q$, i.e.,  $\mathcal{C}_q := \{C \in \mathcal{C}\} | q(C) \leq q\}$. 

Consider the following linear program $LP_{small}$.
\begin{align}
\sum_{C \in \mathcal{C}_{m}}x_{C,|S|+1} & = \eps T \label{eq:lpSmallExtraBox}\\
\sum_{C \in \mathcal{C}_{m_l}}x_{C,l} & = \eps \delta T & \forall l = 1, \dots, |S| \label{eq:lpSmallConfigurationConstraint}\\
\sum_{l = 1}^{|S|} \sum_{C \in \mathcal{C}_{m_l}} x_{C,l}a_{j,C} & = p_j & \forall j \in \bar{\jobs}_{S,W} \label{eq:lpSmallJobConstr}\\
%\sum_{l = 1}^{|S|} \sum_{C \in \mathcal{C}_{m_l}} x_{C}^{(l)}(m_l -q(C)) & \geq W(\jobs_{S,N}) \\
x_{C,l} &\geq 0  & \forall l = 1, \dots, |S|, C \in \mathcal{C}_{m_l}
\end{align}
The variable $x_{C,l}$ defines the processing time of the configuration $C$ in layer $l$.
The condition (\ref{eq:lpSmallConfigurationConstraint}) ensures that we do not give a too large processing time to the configurations used in Layer $l$, while condition (\ref{eq:lpSmallJobConstr}) ensures that the processing time of each job is covered. 
Condition (\ref{eq:lpSmallExtraBox}) is added to place the rounded jobs inside the extra box.
This linear program has $|S| + |\bar{\jobs}_{S,W}|$ conditions and at most $|S||\mathcal{C}_{m}|$ %\leq \frac{(1+2\eps)(1+\eps)}{\eps\delta}(\frac{1}{\eps})^{\log(1/\eps)/\eps}$ 
variables. 
%We can either solve it directly in $\mathcal{O}(..)$ operations or 
If the values $m_l$ are derived from an optimal solution (or are larger than in the corresponding optimal solution), the linear program above has a solution.

To speed up the running time of our algorithm, we do not find a solution to $LP_{small}$. 
Instead, we find a solution to a relaxed version of the linear program, where we allow a slightly increased processing time per layer. 
This linear program is called $LP_{smal, rel}$ and is the same as $LP_{small}$ but we replace equation (\ref{eq:lpSmallConfigurationConstraint}) by 
\[
\sum_{C \in \mathcal{C}_{m_l}}x_{C,l}  = (1+\eps^2)\eps \delta T \ \forall l = 1, \dots, |S|,
\] while equation (\ref{eq:lpSmallExtraBox}) is replaced by 
\[ \sum_{C \in \mathcal{C}_{m}}x_{C,|S|+1} = (1+\eps^2) \eps T/2 \mathrm{.}\]
%\begin{align*}
%\\
%\sum_{l = 1}^{|S|} \sum_{C \in \mathcal{C}_{m_l}} x_{C}^{(l)}a_{j,C} & = p_j & \forall j \in \bar{\jobs}_{S,W}\\
%%\sum_{l = 1}^{|S|} \sum_{C \in \mathcal{C}_{m_l}} x_{C}^{(l)}(m_l -q(C)) & \geq W(\jobs_{S,N}) \\
%x_{C}^{(l)} &\geq 0  & \forall l = 1, \dots, |S|, C \in \mathcal{C}_{m_l}
%\end{align*}

\begin{lemma}
	\label{lma:LPSolveHorizontalJobs}
	If there is a solution to $LP_{small}$, we can find a basic solution to $LP_{small, rel}$ in $\mathcal{O}((|S||\bar{\jobs}_{S,W}|(\ln(|\bar{\jobs}_{S,W}|) +\eps^{-4}))((|S|+|\bar{\jobs}_{S,W}|)^{1.5356} + (\log(1/\eps))^3/\eps^4)) \leq \mathcal{O}(1/\eps^{12}\delta^3)$ operations.
\end{lemma}
\begin{proof}
To solve this linear program, we translate it to a Max-Min-Resource-Sharing problem and solve it with approximation ratio $(1-\rho)$ for $\rho = \mathcal{O}(\eps^2)$ such that $1/(1-\rho) = (1+\eps^2)$.  

In the Max-Min-Resource-Sharing problem, we are given a nonempty convex compact set $B$, and a vector $f$ of $M \in \mathbb{N}$ non-negative continuous concave functions $f:B\rightarrow \mathbb{R}_+^{M}$. 
The objective is to find the value $\lambda^* := \max\sset{\lambda}{f(x)\geq \lambda 1_M, x \in B}$, where $1_M$ is the vector of dimension $M$ with all entries one.
In our translation, we define $f_j := \sum_{l = 1}^{|S|} \sum_{C \in \mathcal{C}_{m_l}} x_{C,l}a_{j,C}/  p_j$ for all 
$j \in \bar{\jobs}_{S,W}$, i.e., $M = |\bar{\jobs}_{S,W}|$ and 
\[B := \sett{x \in \mathbb{R}_{+}^{|S||\mathcal{C}_m|} }{ \sum_{C \in \mathcal{C}_{m_l}}x_{C,l} = \eps\delta T, \forall l = 1, \dots, |S| +1 }\text{.}\] 
We use the algorithm by Grigoriades et al. \cite{GrigoriadisKPV01} to solve this problem. 
This algorithm finds an $x \in B$ that satisfies $f(x) \geq (1-\rho)\lambda^*1_M$. 
To find this solution a so called \textit{approximate block solver} ($\mathcal{ABS}(p,\rho/6)$) has to be provided, where $p\in \mathbb{R}^M_+$.
$\mathcal{ABS}(p,\rho/6)$ has to solve for each $l\in \{1,\dots,|S|\}$ the problem
\begin{align*}
\max \sum_{j \in J}\frac{q_j}{p_j^l}&a_j& \forall j \in \bar{\jobs}_{S,W}\\
\sum_{j \in J}q_j a_j &\leq m_l\\
a_j &\in \mathbb{N}\text{.}
\end{align*}
Intuitively, $\mathcal{ABS}(p,\rho/6)$ computes one configuration for each layer, which is added to the solution $x$ in the next step of the algorithm. 

The above integer program is equivalent to the integer program of the Unbounded Knapsack problem and therefore can be solved approximatively with approximation ratio $(1-\rho/6)$ in $\mathcal{O}(|\bar{\jobs}_{S,W}|+ (\log(1/\rho))^3/\rho^2)$ operations \cite{JansenK18}. 
The algorithm needs at most $\mathcal{O}(M(\ln(M) +\rho^{-2}))$ steps where it calls the $\mathcal{ABS}(p,\rho/6)$ exactly $|S|$ times. Hence the total running time to find $x$ is bounded by 
\begin{align*}
&\mathcal{O}(M(\ln(M) +\rho^{-2})|S|(|\bar{\jobs}_{S,W}|+ (\log(1/\rho))^3/\rho^2)) \\
&= \mathcal{O}(|\bar{\jobs}_{S,W}||S|(\ln(|\bar{\jobs}_{S,W}|) +1/\eps^{4})(|\bar{\jobs}_{S,W}|+ (\log(1/\eps))^3/\eps^4)) \\
&=\mathcal{O}((\log(1/\eps))^3/\eps^{11}\delta),
\end{align*} 
since $|S| = 1/(\eps\delta)$ and $|\bar{\jobs}_{S,W}| = 1/\eps^2$

Note that if the linear program $LP_{small}$ has a solution, there exists an $x'\in B$ with $f_j(x') \geq 1$ for each $j \in \bar{\jobs}_{S,W}$.
However, we solved the Max-Min-Resource-Sharing problem just approximately, i.e., if there exist such an $x'$ with $f_j(x') \geq 1$, it holds for the calculated $x$ that $f_j(x) \geq (1-\rho)$. 
We scale $x$ with $1/(1-\rho)$ and call it $\tilde{x}$. 
If we have that $f_j(\tilde{x}) < 1$ for at least one $j \in \bar{\jobs}_{S,W}$, we know that the liner program $LP_{small}$ has no feasible solution and stop. 
This scaling step extends each layer to $\eps\delta T/(1-\rho) = (1+\eps^2)\eps\delta T$ and therefore it extends the generated schedule by at most $\eps^2T$.

Another obstacle why the given solution $x$ is not a solution to the linear program $LP_{small}$, is that the total reserved processing time for a job $j \in \bar{\jobs}_{S,W}$ in $\tilde{x}$ could be too large, i.e., it could be that
$\sum_{l = 1}^{|S|} \sum_{C \in \mathcal{C}_{m_l}} \tilde{x}_{C,l}a_{j,C} > p_j$ for some $j \in \bar{\jobs}_{S,W}$.
To subduct this surplus, we remove a total processing time of $\sum_{l = 1}^{|S|} \sum_{C \in \mathcal{C}_{m_l}} \tilde{x}_{C,l}a_{j,C} - p_j$ from the configurations for each $j \in \bar{\jobs}_{S,W}$. 
By this step, we create at most one more configuration for each job in $\bar{\jobs}_{S,W}$.
The vector changed in this way, from now on called $\bar{x}$, is a solution to $LP_{small,rel}$.

Since the algorithm in \cite{GrigoriadisKPV01} calls the block solver at most $\mathcal{O}(|S||\bar{\jobs}_{S,W}|(\ln(|\bar{\jobs}_{S,W}|) +\eps^{-4}))$ times, the generated solution $\bar{x}$ uses at most $\mathcal{O}(|S||\bar{\jobs}_{S,W}|(\ln(|\bar{\jobs}_{S,W}|) +\eps^{-4}) +|\bar{\jobs}_{S,W}|)$ configurations in total.
We use the algorithm by Beling and Megiddo \cite{BelingM98} to find a basic solution with at most $|S| +|\bar{\jobs}_{S,W}|$ non zero components in $$\mathcal{O}((|S||\bar{\jobs}_{S,W}|(\ln(|\bar{\jobs}_{S,W}|) +\eps^{-4}))(|S|+|\bar{\jobs}_{S,W}|)^{1.5356}) \leq \mathcal{O}(1/\eps^{12}\delta^3)$$ operations.
Hence the total running time needed to find the basic solution to $LP_{small,rel}$ is bounded by 
$$\mathcal{O}((|S||\bar{\jobs}_{S,W}|(\ln(|\bar{\jobs}_{S,W}|) +\eps^{-4}))((|S|+|\bar{\jobs}_{S,W}|)^{1.5356} + (\log(1/\eps))^3/\eps^4)) \leq \mathcal{O}(1/\eps^{12}\delta^3).$$
This concludes the prove.  
\end{proof}
We will find a schedule of the jobs $\jobs_{S,W}$, by placing the configurations into the corresponding layers and greedily filling the jobs into the configurations, see Figure \ref{fig:placingOfSmallItems}.
To ensure that each job can be scheduled integrally, we extend each configuration, by $\mu T$, which is the tallest height a small job can have. Since there are at most $|S| + |\bar{\jobs}_{S,W}|$ configurations we extend the schedule by at most $(|S| + |\bar{\jobs}_{S,W}|)\mu T \leq (1/\eps\delta + 1/\eps^2)\delta \eps^3T\leq 2\eps^2T$.
Note that after this extension the size defining job, which might has been cut for the analysis, can be scheduled in the group where it first appears. 

To schedule the small jobs, we use the next fit decreasing height (NFDH) algorithm to place them next to the configurations. 
We can sort the small jobs by height in $\mathcal{O}(n+\log(n)/\eps^2)$ since there are at most $\mathcal{O}(\log(n)/\eps^2)$ possible processing times.

Note that the total work of the small jobs has to fit next to the configurations. The reason is that the configurations have a total work which equals the total work of the wide jobs. 
Furthermore after scheduling the large jobs, the total idle time of the machines was at least as large as the total work of the small jobs. 
%Otherwise we wold not have tried this linear program in the first place. 

\begin{figure}[ht]
	\centering
	%\resizebox{0.45\textwidth}{!}{%
	\begin{tikzpicture}
	\pgfmathsetmacro{\h}{0.65}
	\pgfmathsetmacro{\w}{0.55}
	\draw (0*\w,6*\h) -- (0*\w,0*\h) -- (10*\w,0*\h) -- (10*\w,6*\h);
	
	\draw[fill = gray] (0*\w,0.0*\h) rectangle (4.0*\w,0.5*\h);
	\draw[fill = gray] (0*\w,0.5*\h) rectangle (3.9*\w,0.9*\h);
	\draw[fill = gray] (0*\w,0.9*\h) rectangle (3.8*\w,1.4*\h);
	\draw[fill = gray] (0*\w,1.4*\h) rectangle (3.7*\w,1.7*\h);
	
	\draw[fill = gray] (4*\w,0.0*\h) rectangle (7.5*\w,0.4*\h);
	\draw[fill = gray] (4*\w,0.4*\h) rectangle (7.4*\w,0.8*\h);
	\draw[fill = gray] (4*\w,0.8*\h) rectangle (7.3*\w,1.3*\h);
	\draw[fill = gray] (4*\w,1.3*\h) rectangle (7.2*\w,1.8*\h);
	
	%small
	\draw[fill = lightgray] (7.5*\w,0*\h) rectangle (8.5*\w,1.0*\h);
	\draw[fill = lightgray] (8.5*\w,0*\h) rectangle (9.4*\w,0.9*\h);
	\draw[fill = lightgray] (7.5*\w,1*\h) rectangle (8.3*\w,1.9*\h);
	\draw[fill = lightgray] (8.3*\w,1*\h) rectangle (9.3*\w,1.8*\h);
	
	\draw[ pattern= north west lines] (7.5*\w,1*\h) rectangle (10*\w,0.9*\h);
	
	\draw[|-|] (11*\w,1*\h) -- (11*\w,0.9*\h);
	
	\draw[] (11*\w,0.95*\h) -- (12*\w,4*\h);
	
	\draw[fill = gray] (0*\w,2.5*\h) rectangle (3.1*\w,3.4*\h);
	\draw[fill = gray] (0*\w,3.4*\h) rectangle (2.9*\w,3.9*\h);
	\draw[fill = gray] (0*\w,3.9*\h) rectangle (2.8*\w,4.5*\h);
	
	\draw[fill = gray] (3.2*\w,2.5*\h) rectangle (6.1*\w,3.1*\h);
	\draw[fill = gray] (3.2*\w,3.1*\h) rectangle (6.0*\w,3.8*\h);
	\draw[fill = gray] (3.2*\w,3.8*\h) rectangle (5.9*\w,4.7*\h);
	
	%small
	\draw[fill = lightgray] (6.1*\w,2.5*\h) rectangle (7.0*\w,3.3*\h);
	\draw[fill = lightgray] (7.0*\w,2.5*\h) rectangle (8.0*\w,3.2*\h);
	\draw[fill = lightgray] (8.0*\w,2.5*\h) rectangle (8.5*\w,3.1*\h);
	\draw[fill = lightgray] (8.5*\w,2.5*\h) rectangle (9.4*\w,3.0*\h);
	
	\draw[ pattern= north west lines] (6.1*\w,3.3*\h) rectangle (10*\w,3.0*\h);
	\draw[|-|] (11*\w,3.3*\h) -- (11*\w,3.0*\h);
	\draw[] (11*\w,3.15*\h) -- (12*\w,4*\h);
	
	\draw[fill = lightgray] (6.1*\w,3.3*\h) rectangle (6.9*\w,3.8*\h);
	\draw[fill = lightgray] (6.9*\w,3.3*\h) rectangle (7.9*\w,3.8*\h);
	\draw[fill = lightgray] (7.9*\w,3.3*\h) rectangle (8.7*\w,3.7*\h);
	\draw[fill = lightgray] (8.7*\w,3.3*\h) rectangle (9.6*\w,3.7*\h);
	
	\draw[ pattern= north west lines] (6.1*\w,3.7*\h) rectangle (10*\w,3.8*\h);
	\draw[|-|] (11*\w,3.7*\h) -- (11*\w,3.8*\h);
	\draw[] (11*\w,3.75*\h) -- (12*\w,4*\h);
	
	\draw[fill = lightgray] (6.1*\w,3.8*\h) rectangle (6.8*\w,4.2*\h);
	\draw[fill = lightgray] (6.8*\w,3.8*\h) rectangle (7.8*\w,4.2*\h);
	\draw[fill = lightgray] (7.8*\w,3.8*\h) rectangle (8.6*\w,4.2*\h);
	\draw[fill = lightgray] (8.6*\w,3.8*\h) rectangle (9.2*\w,4.2*\h);

	\draw[fill = lightgray] (0.0*\w,5*\h) rectangle (1.0*\w,5.4*\h);
	\draw[fill = lightgray] (1.0*\w,5*\h) rectangle (1.9*\w,5.35*\h);
	\draw[fill = lightgray] (1.9*\w,5*\h) rectangle (2.9*\w,5.35*\h);
	\draw[fill = lightgray] (2.9*\w,5*\h) rectangle (3.8*\w,5.3*\h);
	\draw[fill = lightgray] (3.8*\w,5*\h) rectangle (4.7*\w,5.3*\h);
	\draw[fill = lightgray] (4.7*\w,5*\h) rectangle (5.7*\w,5.3*\h);
	\draw[fill = lightgray] (5.7*\w,5*\h) rectangle (6.6*\w,5.3*\h);
	\draw[fill = lightgray] (6.6*\w,5*\h) rectangle (7.5*\w,5.25*\h);
	\draw[fill = lightgray] (7.5*\w,5*\h) rectangle (8.5*\w,5.25*\h);
	\draw[fill = lightgray] (8.5*\w,5*\h) rectangle (9.5*\w,5.25*\h);
	
	\draw[ pattern= north west lines] (0*\w,5.25*\h) rectangle (10*\w,5.4*\h);
	\draw[|-|] (11*\w,5.25*\h) -- (11*\w,5.4*\h);
	\draw[] (11*\w,5.325*\h) -- (12*\w,4*\h) node[right]{$\sum \leq \mu T$};
	
	\draw[fill = lightgray] (0.0*\w,5.4*\h) rectangle (0.9*\w,5.65*\h);
	\draw[fill = lightgray] (0.9*\w,5.4*\h) rectangle (1.8*\w,5.65*\h);
	\draw[fill = lightgray] (1.8*\w,5.4*\h) rectangle (2.8*\w,5.65*\h);
	\draw[fill = lightgray] (2.8*\w,5.4*\h) rectangle (3.8*\w,5.6*\h);
	\draw[fill = lightgray] (3.8*\w,5.4*\h) rectangle (4.8*\w,5.6*\h);
	\draw[fill = lightgray] (4.8*\w,5.4*\h) rectangle (5.8*\w,5.6*\h);
	\draw[fill = lightgray] (5.7*\w,5.4*\h) rectangle (6.7*\w,5.6*\h);
	\draw[fill = lightgray] (6.7*\w,5.4*\h) rectangle (7.6*\w,5.55*\h);
	\draw[fill = lightgray] (7.6*\w,5.4*\h) rectangle (8.6*\w,5.55*\h);
	\draw[fill = lightgray] (8.6*\w,5.4*\h) rectangle (9.6*\w,5.55*\h);
	
	\draw[dashed, very thick] (0,1.5*\h) -- (10*\w,1.5*\h);
	\draw[dashed, very thick] (0,2.5*\h) -- (10*\w,2.5*\h);
	\draw[dashed, very thick] (0,4.0*\h) -- (10*\w,4.0*\h);
	\draw[dashed, very thick] (0,5.0*\h) -- (10*\w,5.0*\h);\draw[dashed, very thick] (0,5.45*\h) -- (10*\w,5.45*\h);

	\draw [decorate,decoration={brace,amplitude=5pt},thick] (10*\w,0*\h) -- (9*\w,0*\h) node [midway,yshift=-10]{$< \eps m$};
	\draw[ pattern= dots] (9*\w,0.0*\h) rectangle (10*\w,1.50*\h);
	\draw[ pattern= dots] (9*\w,2.5*\h) rectangle (10*\w,4.00*\h);
	\draw[ pattern= dots] (9*\w,5.0*\h) rectangle (10*\w,5.45*\h);
	
	\draw [decorate,decoration={brace,amplitude=5pt},thick] (10*\w,2.5*\h) -- (10*\w,1.5*\h) node [midway,xshift=15]{$\mu T$};

	\end{tikzpicture}
	%}
	\caption{Filled configurations, containing wide (dark gray) and narrow (light gray) small jobs}
	\label{fig:placingOfSmallItems}
\end{figure}
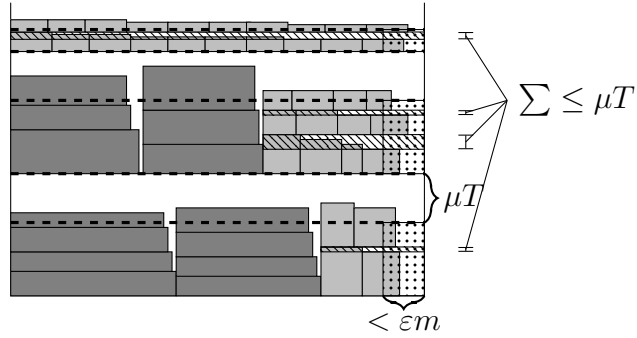

The NFDH algorithm sorts the small jobs by height and places them into shelves starting with the tallest job, see Figure~\ref{fig:placingOfSmallItems}.
In each shelf there are at most $\eps m$ machines which are completely idle since each narrow job requires at most $\eps m$ machines. 
If there would be more idle machines, another job would have fitted in this shelf.

Furthermore, there can be machines that start to idle before the starting time of the next shelf, namely in the moment when a job with a processing time smaller than the first job in this shelf has finished its processing time. 
Let $p_{\max,i}$ be the largest processing time in shelf $i$, then the idle time of the machines which start to idle in shelf $i$ is bounded by $p_{\max, i} - p_{\max,i+1}$. Therefore in total the processing time of machines starting to idle over all shelves is bounded by $p_{\max}\cdot m$ while the total idle time of machine being idle during the whole shelf is bounded by $\eps m \cdot T$. 
Hence the total work of narrow small jobs that cannot be scheduled next to the configurations, is bounded by $\eps m \cdot T + \mu T \cdot m$.
By using NFDH again to schedule these jobs, we add at most  $(\mu+\eps) m T/(1-\eps)m +p_{\max} \leq 2\eps T$ to the makespan. 

\subsection{Medium Jobs}
In the last step, we schedule the medium sized jobs. 
First, we sort them by their processing time. This can be done in $\mathcal{O}(n + 1/\eps^2)$ since there are at most $3/\eps^2$ different processing times between $\mu = \eps^3\delta$ and $\delta$. 
Afterward, we use the NFDH algorithm to place the jobs. 
Hence, we start with the tallest one and place the jobs one by one in shelves. %, see Figure \ref{fig:NFDH}. 
Coffman et al. \cite{CoffmanGJT80} have shown the following (slightly adapted) Lemma:

\begin{lemma}[See \cite{CoffmanGJT80}]
	\label{thm:NFDH}
	For any list $L$ ordered by nonincreasing height, 
	\[\mathrm{NFDH}(L) \leq 2W(L)/m + p_{max} \leq 2\cdot \OPT(L) + p_{max}\text{.}\]
\end{lemma}

We know that $W(\jobs_M)$ is bounded by $\eps m T$ and $p_{\max}$ is bounded by $\delta T \leq \eps T$.
Therefore, we add at most $\mathrm{NFDH}(\jobs_M) \leq 2\eps T + \eps T \leq 6 \eps \OPT$ to the makespan, by scheduling the medium sized jobs this way.

\subsection{Summary}
\label{subsec:summary}
Given $\jobs$, $m$ and $\eps$ the algorithm can be summarized as follows:

\begin{enumerate}%[itemsep =0pt,parsep=0pt,partopsep=0pt]
	\item In the first step of the algorithm, we simplify the instance. 
	We define the lower bound $T:= \max\{p_{\max}, (\sum_{j \in \jobs}p_jq_j)/m\}$ and round the processing times such that they are multiples of $\eps T/n$.
	%Then, we scale the instance with $n/(\eps T)$ such that each job has a processing time in $\{1, \dots,n/\eps\}$. 
	Next, we find the correct values for $\delta$ and $\mu$ and partition the jobs into $\jobs_{L,W} \disjCup \jobs_{L,N} \disjCup \jobs_M \disjCup  \jobs_{S,W}\disjCup\jobs_{S,N}$ accordingly.
	Afterward, we round the processing times of all jobs using Lemma \ref{lma:roundingProcessingTimes} and generate $\bar{\jobs}_{L,N}$.
	Last, we generate $\bar{\jobs}_{S,W}$, i.e., we round the machine requirements of the horizontal jobs.
	\item After the simplification steps area done, we start a binary search for the correct size of $\OPT$ for the rounded instance. Note that to find the correct value for $\OPT$ for the rounded instance, we are only interested in the number of layers needed to place the jobs in $\jobs_{L,W} \disjCup \bar{\jobs}_{L,N}\disjCup \bar{\jobs}_{S,W}$.
	We know that we need at least $l= 1/(\eps\delta)$ layer but at most $u = 2(1+\eps)(1+2\eps)/(\eps\delta)$ layer.
	We start our binary search using $L = \lfloor(l + u)/2\rfloor$ layers.
	\item Given a number of layers $L$, we try each possibility to schedule $\jobs_{L,W} \cup \bar{\jobs}_{L,N}$ using at most this number of layers. For each of these possibilities, we try to solve $LP_{small}$ for $\bar{\jobs}_{S,W}$ with last allowed layer $L$. If $LP_{small}$ is solvable, we save the $LP$-solution and the choice for $\jobs_{L,W} \cup \bar{\jobs}_{L,N}$ and set the upper bound $u = L-1$ update $L$ accordingly.
	Otherwise, we try the next choice for $\jobs_{L,W} \cup \bar{\jobs}_{L,N}$. If all the possibilities to schedule these jobs fail, we set $l = L+1$ and update $L$ accordingly.
	\item The binary search part is finished as soon as $u < l$.
	When this is the case, we consider the last solution for $LP_{small}$ and the corresponding choice for $\jobs_{L,W}$ and $\bar{\jobs}_{L,N}$.
	we scale back all the processing times and assign the jobs $\jobs_{L,N}$ and $\jobs_{S,W}$ to this solution and shift the fractional placed jobs to the top or the extra cluster $C_2$ as described in Section \ref{sec:LargeJobs} and Section \ref{sec:smallJobs}.
	We schedule the jobs $\jobs_{S,N}$ next to the configurations for $\bar{\jobs}_{S,W}$ using the NFDH algorithm. 
	Finally, we schedule the medium sized jobs on top of the schedule using the NFDH algorithm.
\end{enumerate}

In most of the simplification steps, we have some loss in the approximation ratio of size $\mathcal{O}(\eps T)$. Since $T \leq \OPT$ it holds that the algorithm has an approximation ratio of the form $(1+\Oh(\eps))\OPT + p_{\max}$.
To reach a $(1+\eps)\OPT + p_{\max}$ algorithm the value $\eps$ has to be scaled accordingly.
Note that for the sake of simplicity, we did not optimize the above algorithm to guarantee the best possible running time with regard to the added $\Oh(\eps)$. 

The total running time of the algorithm is bounded by 
\begin{align*}
&\mathcal{O}(\log(n)/\eps^2+n \log(1/\eps)
+\log(1/\eps\delta)((|S|^{1/\delta\alpha} (|S||P_L|/\alpha)^{|S|+|P_L|}))(1/\eps^{10}\delta^2)) 
\\
=& \mathcal{O}(n \log(1/\eps))
+ 1/\eps^{1/\eps^{\mathcal{O}(1/\eps)}},
\end{align*}
which concludes the proof of Theorem \ref{thm:AEPTASPTS}. 

Next, we describe how this leads to the algorithm from Theorem \ref{thm:ONAlg}.
As described, we can use this algorithm to find the schedule on the clusters $C_1$ and $C_2$ as needed for the algorithm described in the proof of Lemma \ref{lma:partitioningLemma}. 
Both algorithms combined have the properties of the algorithm needed for the proof of Theorem \ref{thm:ONAlg}.
The algorithm from Lemma \ref{lma:partitioningLemma} will call the above algorithm with $\eps = 1/8$ if $N = 2$ or $\eps \geq 1/4$ for the case that $N \geq 6$. 
Hence in the worst case the additive constant becomes something like $\Omega(8^{8^8})$ for $N = 2$  and $\Omega(4^{256})$ for $N \geq 6$.
However, note that the above running time is a worst case running time and that, depending on the instance, we might have $\delta = \eps$ rather than $\delta = \eps^{3/\eps}$, what will reduce the additive constant significantly.

To prove the \ac{MCS} part of Theorem \ref{thm:AEPTASMCS} note that we can use the algorithm described in this section to find a schedule on $N$ clusters with ratio $(1+\eps)\OPT + p_{\max}$ in the same running time. 
Let $\OPT$ be the makespan of an optimal schedule on $N\geq 2$ clusters.
Consider a solution for an instance of \acl{PTS} generated by the algorithm above.
It has a makespan of $T_{\schedAlg}\leq (1+\eps)N\OPT +p_{\max}$.
Define $T_{\schedAlg}' := T_{\schedAlg} -p_{\max}$.
We partition the schedule at multiples of $T_{\schedAlg}'/N$, and schedule each job starting between two of these multiples on the same cluster, such that the jobs remain their relative starting positions. 
Since $T_{\schedAlg}'\leq (1+\eps)N\OPT$, each of these parts has a height of at most $(1+\eps)\OPT +p_{\max}$.
This concludes the proof of Theorem \ref{thm:AEPTASMCS}.

\section{A Faster Algorithm for a Practical Number of Jobs}
\label{sec:FastAlgorithm}
Note that in the algorithm described above, we have a running time of $\mathcal{O}(n)$, but the hidden constant can be extremely large. 
Hence, in practical applications it can be more useful to use an algorithm with running time $\mathcal{O}(n\log(n))$ or $\mathcal{O}(n^2)$, to find an $\alpha\OPT+p_{\max}$ approximation for \acf{PTS}.
For $N \geq 6$, we use $\eps \in [1/4,1/3]$ 
%since $\lfloor N/3\rfloor/N \in [1/4,1/3]$ for $N \geq 6$ 
and, hence, a fast $\mathrm{poly}(n)$ algorithm without large hidden constants and  approximation ratio $(5/4)\OPT + p_{\max}$ would bring an significant improvement for the vast majority of cluster numbers with $2$ and $5$ being the only exceptions. 
Even an algorithm with approximation ratio $(4/3)\OPT +p_{\max}$ would speed up the algorithm for one third of all the possible instances, namely all the instances where the number of clusters is dividable by three.

To this point, we did not find either of the algorithms, and we leave this as an open question. 
Instead, we present a fast algorithm with approximation ratio $(3/2)\OPT + p_{\max}$. 
This algorithm for \ac{PTS} leads to an algorithm for \ac{MCS} with approximation ratio $9/4$ for all instances where $N \bmod 3 =0$.

In the description of the following algorithm, we need the concept of an idle machine.
A machine is \textit{idle} at a time $\tau$ if it does not processes any job at that time. 
Given a point in time $\tau$ the number of idle machines at that time is given by \[\free(\tau) := m - \sum_{\substack{j \in \jobs,\\ \sigma(j) \leq \tau < \sigma(j)+\pT{j},\\ \rho(j) = i}}\mN{j}\] and the total idle time up to $\tau$ is defined by 
\[\tau m - \sum_{\substack{j \in \jobs,\\  \sigma(j) +\pT{j} \leq \tau}} \mN{j}\pT{j}  + \sum_{\substack{j \in \jobs,\\  \sigma(j) \leq \tau < \sigma(j)+\pT{j}}} \mN{j}(\tau-\sigma(j)).\]

\begin{lemma}
	\label{lma:fastPTS}
	%Define $T := \max\{\work(\jobs),  \pT{\sset{j \in \jobs}{\mN{j}>m/2}}\}$.
	There is an algorithm for \ac{PTS} with approximation guarantee $ (3/2)\OPT + p_{\max}$ and running time $\mathcal{O}(n\log(n))$. 
	This schedule can be divided into two clusters $C_1$ and $C_2$, where the schedule on $C_1$ has a makespan of at most $(3/2) \OPT$ and the makespan of $C_2$ is bounded by $p_{\max}$.
\end{lemma}

\begin{proof}
	%First, note that $T$ is obviously a lower bound on the optimum since it is not possible to schedule two jobs from the set $\sset{j \in \jobs}{\mN{j}>m/2}$ at the same time.
	In the following, we describe the steps of the algorithm. 
	The first part of the algorithm is to find a schedule for the jobs with machine requirement larger than $m/3$.
 	In the second part, we schedule the jobs with machine requirement at most $m/3$ in a best fit manner.
 	This second part depends on one property from the schedule for the jobs with resource requirement larger than $m/3$, as we will see later.
 	This algorithm uses the following optimized variant of List-Scheduling as described in Turek et al~\cite{TurekWY92}: 
 	Starting at time $\tau = 0$ for every endpoint of a job, schedule the widest job that can be started at this point if there is one; otherwise, go to the next endpoint and proceed as before.
	
	The first part of the algorithm can be summarized as follows:
	\begin{enumerate}%[noitemsep]%[leftmargin=3ex,labelindent=1ex,labelsep = 0.0pt,itemsep=0.0\baselineskip,topsep=0.4\baselineskip,partopsep=0ex, parsep= 0em, listparindent=\parindent,style = sameline]%[itemsep=0em]%[topsep=0em, partopsep=0em, parsep=0em, itemsep=0em]
	\item For a given set of jobs $\jobs$, first consider the jobs $j\in\jobs$ with $\mN{j} \in [m/3, m]$ and sort them by decreasing size of the machine requirement $\mN{j}$. 
	\item We stack all the jobs $j \in \jobs$ with $\mN{j} >m/2$ ordered by their machine requirement such that the largest starts at time $0$, see Figure~\ref{fig:HeuristicalAlgorithm}.
	\item Look at the job with the smallest requirement of machines larger than $m/3$ and place it at the first possible point in the schedule next to the jobs with machine requirement larger than $m/2$. 
	We call this point in time $\tau$. 
	\item Schedule all the other jobs with machine requirement at least $m/3$ with the optimized List-Schedule starting at $\tau$. The List-Schedule includes the endpoints of the already scheduled jobs.
	\end{enumerate}
	
	Let $\tau'$ be the point in time, where the last job ends, which needs more than $m/2$ machines and define $\tau''$ to be the first point in time where both jobs scheduled at $\tau'$ have ended.
	Furthermore, let $T'$ be the last point in the schedule where two jobs are processed and define $T := \max\{T',\tau'\}$.
	Note that at each point in the schedule between $\tau'$ and $T$ there will be scheduled exactly two jobs with machine requirement in $[m/3,m/2]$, while between $\tau'$ and $\tau$ it can happen that there is no job from this set. 
	
	\begin{figure}[ht]
		\begin{center}
			\begin{tikzpicture}
			\pgfmathsetmacro{\h}{0.9}
			\pgfmathsetmacro{\w}{5}
			
			\draw[] (0,4*\h)--(0,0) -- (\w,0) -- (\w,4*\h);
			\draw[dashed] (\w/2,-0.4*\h) node[below]{$m/2$} -- (\w/2,4*\h);
			\draw[dashed] (\w/3,-0.1*\h) node[below]{$m/3$} -- (\w/3,4*\h);
			\draw[dashed] (2*\w/3,-0.1*\h) node[below]{$\frac{2}{3}m$} -- (2*\w/3,4*\h);
			
			\draw[fill=gray,fill opacity=0.3] (0.1*\w,0.0*\h) rectangle (\w,0.2*\h); 
			\draw[fill=gray,fill opacity=0.3] (0.2*\w,0.2*\h) rectangle (\w,0.5*\h); 
			\draw[fill=gray,fill opacity=0.3] (0.3*\w,0.5*\h) rectangle (\w,0.7*\h); 
			%\draw[dashed] (-0.15*\w,0.7*\h) -- (\w,0.7*\h);
			\draw[fill=gray,fill opacity=0.3] (0.35*\w,0.7*\h) rectangle (\w,1.0*\h);
			\draw[fill=gray,fill opacity=0.3] (0.38*\w,1.0*\h) rectangle (\w,1.2*\h);
			\draw[fill=gray,fill opacity=0.3] (0.42*\w,1.2*\h) rectangle (\w,1.4*\h);
			\draw[fill=gray,fill opacity=0.3] (0.46*\w,1.4*\h) rectangle (\w,1.6*\h);
			\draw[fill=gray,fill opacity=0.3] (0.48*\w,1.6*\h) rectangle (\w,1.7*\h);
			\draw[dashed] (\w,1.7*\h) --(1.1*\w,1.7*\h) node[right]{$\tau'$};
			
			\draw[dashed] (-0.1*\w,0.7*\h) node[left]{$\tau$} -- (\w,0.7*\h);
			\draw[fill=gray,fill opacity=0.3] (0*\w,0.7*\h) rectangle (0.34*\w,0.9*\h);
			\draw[fill=gray,fill opacity=0.3] (0*\w,1.2*\h) rectangle (0.39*\w,1.5*\h);
			\draw[fill=gray,fill opacity=0.3] (0*\w,1.5*\h) rectangle (0.46*\w,1.8*\h);
			
			\draw[fill=gray,fill opacity=0.3] (0.5*\w,1.7*\h) rectangle (\w,2.0*\h);
			\draw[fill=gray,fill opacity=0.3] (0.52*\w,2.0*\h) rectangle (\w,2.3*\h);
			\draw[fill=gray,fill opacity=0.3] (0.54*\w,2.3*\h) rectangle (\w,2.5*\h);
			\draw[fill=gray,fill opacity=0.3] (0.56*\w,2.5*\h) rectangle (\w,2.7*\h);
			\draw[fill=gray,fill opacity=0.3] (0.58*\w,2.7*\h) rectangle (\w,3.0*\h);
			\draw[fill=gray,fill opacity=0.3] (0.60*\w,3.0*\h) rectangle (\w,3.3*\h);
			\draw[dashed](-0.1*\w,3.3*\h) node[left]{$T$} -- (1.1*\w,3.3*\h);
			
			\draw[fill=gray,fill opacity=0.3] (0*\w,1.8*\h) rectangle (0.49*\w,2.1*\h);
			\draw[fill=gray,fill opacity=0.3] (0*\w,2.1*\h) rectangle (0.47*\w,2.4*\h);
			\draw[fill=gray,fill opacity=0.3] (0*\w,2.4*\h) rectangle (0.45*\w,2.6*\h);
			\draw[fill=gray,fill opacity=0.3] (0*\w,2.6*\h) rectangle (0.43*\w,2.9*\h);
			\draw[fill=gray,fill opacity=0.3] (0*\w,2.9*\h) rectangle (0.41*\w,3.2*\h);
			\draw[fill=gray,fill opacity=0.3] (0*\w,3.2*\h) rectangle (0.39*\w,3.5*\h);
			
			%\draw[|-|] (1.05*\w, 0) -- node[midway,right]{$a$} (1.05*\w,1.7*\h);
			%\draw[|-|] (1.05*\w, 1.7*\h) -- node[midway,right]{$b$} (1.05*\w,3.3*\h);
			
			\draw[|-|] (-0.05*\w,  0.0*\h) -- (-0.05*\w,0.7*\h);
			\draw[|-|] (-0.05*\w, 0.7*\h) -- (-0.05*\w,0.9*\h);
			\draw[|-|] (-0.05*\w,  0.9*\h) -- (-0.05*\w,1.2*\h);
			\draw[|-|] (-0.05*\w, 1.2*\h) --  (-0.05*\w,3.3*\h);
			
			\coordinate (A) at (-0.4*\w, 0.4*\h);
			\node[left] at (A) {a};
			\coordinate(B) at (-0.4*\w, 2.0*\h);
			\node[left] at (B) {b};
			
			\coordinate (C) at (-0.05*\w, 0.5*0.0*\h + 0.5*0.7*\h);
			\coordinate (D) at (-0.05*\w, 0.5*0.7*\h + 0.5*0.9*\h);
			\coordinate (E) at (-0.05*\w, 0.5*0.9*\h + 0.5*1.2*\h);
			\coordinate (F) at (-0.05*\w, 0.5*1.2*\h + 0.5*3.3*\h);
			\draw[gray] (A) -- (C);
			\draw[gray] (A) -- (E);
			\draw[gray] (B) -- (D);
			\draw[gray] (B) -- (F);
			\end{tikzpicture}
		\end{center}
		\caption{A placement of the jobs with processing time larger than $m/3$.}
		\label{fig:HeuristicalAlgorithm}
	\end{figure}
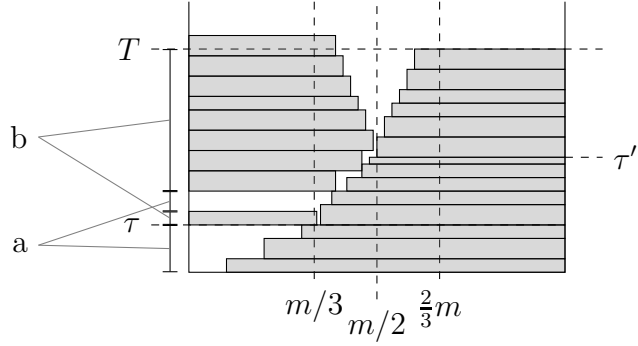
	
	We claim that $T \leq \OPT$. 
	If $T = \tau'$, this is obvious since we never can schedule jobs with machine requirement larger than $m/2$ at the same time. 
	Consider the case that $T = T'$. 
	Between $\tau'$ and $T$ there will be scheduled two jobs at each point in the processing time. 
	Hence if $T$ is larger than $\OPT$ there has to be a schedule, where one more of the jobs between this points of time is scheduled below $\tau'$. 
	But since each job is scheduled as early as possible, there can be no such job, which proves the claim.
	
	In the next step, we are going to schedule the residual jobs, which have a machine requirement of at most $m/3$. 
	In order to schedule these jobs, we might reconstruct the schedule generated so far.
	This reconstruction is necessary, if the schedule generated so far has a to large amount of idle time on the machines. 
	As a result of this large amount of idle time, we cannot guarantee a small approximation ratio, when scheduling the residual jobs. 
	Furthermore, note that if there are no jobs with machine requirement at most $m/3$, we do not need to add further steps and have found a schedule with approximation guarantee $\OPT + p_{\max}$.
	
	Let $a$ be the total processing time before $T$, where only one job is scheduled. This job has to be a job with machine requirement larger than $m/2$. 
	Let $b$ be the total processing time, where just two jobs are scheduled. 
	We will now consider two cases: $a>b$ and $a \leq b$.
	In the first case, we have to reconstruct the schedule found so far, while in the second case this is not necessary.
	
	We can summarize the second part of the algorithm, where we schedule the jobs with machine requirement at most $m/3$, as follows:
	\begin{enumerate}[resume]%[noitemsep]%[before={\setcounter{steplistcount}{4}},leftmargin=1ex,labelindent=1ex,labelsep = 0.0pt,itemsep=0.0\baselineskip,topsep=0.4\baselineskip,partopsep=0ex, parsep= 0em, listparindent=\parindent,style = sameline]
			\item Find $a$ and $b$
		\item If $a > b$,
		dismantle the schedule and stack all the jobs with machine requirement larger than $m/3$ on top of each other, sorted by machine requirement such that the widest one starts at $0$.
		Schedule the residual jobs with the modified List-Schedule starting at $0$ and using the endpoints of all jobs.
		\item Else if $a \leq b$, determine $\tau''$ and use the optimized List-Schedule to schedule the remaining starting at $\tau''$ while using the endpoints of all scheduled jobs.
	\end{enumerate}
	
	In the following, we will argue that the second part of the described algorithm delivers a schedule with approximation guarantee $(3/2)\OPT + p_{\max.}$
	%\textit{The case $a>b$:} 
	\begin{caseList}
	%\subparagraph*{Case 1: $a>b$}
	\item[$a>b$.]
	In this case, the algorithm performs the following steps: 
	We stack all the jobs with machine requirement larger than $m/3$ on top of each other sorted by decreasing number of required machines. 
	This stack has a height of at most $a+2b+p_{\max}$ and the last job of this stack is starting before or at $a+2b$.
	For the remaining jobs, i.e., the jobs with machine requirement at most $m/3$, we us the the improved List-Schedule algorithm as described in Turek et al. \cite{TurekWY92}.
	This means, we go through the schedule from the bottom to the top and look for each end point of jobs $t$, starting with $t = 0$, at the number of idle machines $\free(t)$. 
	We search for the widest unscheduled job with $\mN{j} \leq \free(t)$ and start it at this time, if one exists, and calculate the new number of idle machines at this point in time. 
	If no such job exists, we go to the next end point of a job since the number of idle machines only changes at these points.
	
	We claim that this schedule has a makespan of at most $(3/2)\OPT +p_{\max}$.
	Let $\rho$ be the last starting point of a job in this schedule. 
	If this point is larger than $a+2b$, the last scheduled job has a machine requirement of at most $m/3$. 
	By construction of the schedule, this job could not be scheduled at any earlier time. 
	Hence at each time in the schedule before $\rho$, we use at least $(2/3)m$ machines and therefore $\rho (2/3)m \leq W(\jobs)$. 
	Furthermore, we know that $\OPT \geq W(\jobs)/m$. 
	As a consequence it holds that $\rho \leq (3/2)\OPT$. 
	Since $\rho$ is the last starting position of all jobs, the makespan of the schedule is bounded by $(3/2)\OPT +p_{\max}$.
	
	On the other hand if $\rho \leq a+2b$, the last starting job can be a job with machine requirement larger than $m/3$. 
	However, the schedule is then bounded by $a+2b+p_{\max}$.
	Since $a>b$ and $a+b \leq \OPT$ it holds that $b \leq \OPT/2$ and therefore $a+2b+p_{\max} \leq (3/2)\OPT + p_{\max}$.
	
	%\textit{The case $a\leq b$:} 
	%\subparagraph*{Case 2: $a\leq b$}
	\item[$a\leq b$.]
	We now consider the case that $a \leq b$. 
	In this scenario, we do not dismantle the given schedule as we do in Case 1. 
	Instead, we use the improved List-Schedule algorithm as described in Turek et al. \cite{TurekWY92} to schedule the remaining jobs.
	To prove that the resulting algorithm has an approximation guarantee of $(3/2)\OPT + p_{\max}$, we analyze the total idle time up to the point $T$ before we schedule the residual jobs.
	
	Let $t_a$ be an arbitrary point in time before $T$ where only one job is scheduled and let $t_b$ be an arbitrary point in time where two jobs are scheduled. 
	Note that $\free(t_b) < m/3$ since both jobs scheduled at this time have a machine requirement of at least $m/3$.
	We differentiate two cases $t_b < \tau'$ and $t_b \geq \tau'$ and claim that in both cases the sum of numbers of idle machines at $t_a$ and $t_b$ is bounded by $\frac{2}{3}m$. 
	As a a consequence of this claim, the average number of idle machines at all of these pairs of points is bounded by $m/3$ and hence, the total idle time up tp the point $T$ is bounded by $m/3 \cdot (a+b)\leq T m/3$ because $a \leq b$ and at each point $t_b$ the idle time is bounded by $m/3$.
	
	%\begin{claim}
	%	The sum of idle times at $t_a$ and $t_b$ is bounded by $2m/3$.
	%\end{claim}
	\begin{caseList}
	%\subparagraph*{Case 2.1: $t_b < \tau'$}
	\item[$t_b < \tau'$.]
	In the case that $t_b < \tau'$, the number of idle machines $\free(t_b)$ is bounded by $m/6$ since there is scheduled one job with machine requirement at least $m/2$ and one job with machine requirement at least $m/3$. 
	On the other hand, $\free(t_a)$ is bounded by $m/2$. 
	Therefore, the sum of free machines on both points is bounded by $\frac{2}{3}m$ and hence the average is bounded by $m/3$.
	
	%\subparagraph*{Case 2.2: $t_b \geq \tau'$}
	\item[$t_b \geq \tau'$.]
	If $t_b \geq \tau'$, there are two jobs with machine requirement at least $m/3$ scheduled at this point in time and hence $\free(t_b) < m/3$. 
	Remember that $t_a < \tau'$ since at each point in time after the point $\tau'$ up to the point $T$ there will be two jobs scheduled.
	Therefore, $t_a < t_b$ and the jobs scheduled at $t_b$ did not fit at the time $t_a$ since otherwise they would have been scheduled there. 
	As a consequence, it holds that $\free(t_a) \leq (m-\free(t_b))/2$ because the job with the smaller machine requirement scheduled at $t_b$ has a machine requirement of at most $(m-\free(t_b))/2$. 
	Hence it holds that $\free(t_a)+\free(t_b)\leq m/2+\free(t_b)/2$.
	Since $\free(t_b) \leq m/3$, we have $\free(t_a)+\free(t_b)\leq \frac{2}{3}m$.
	\end{caseList}
	
	In conclusion, we have $\free(t_a)+\free(t_b)\leq \frac{2}{3}m$ in both cases $t_b < \tau'$ and $t_b \geq \tau'$. 
	Hence the average number of idle machines for each pair of two points $t_a$ and $t_b$ is bounded by $m/3$.
	Since $a \leq b$ and at each point $t_b$, where two jobs are scheduled, there are at most $m/3$ machines idle, the total idle time below $T$ is bounded by $T m/3$. 
	The residual jobs are scheduled by the best fit algorithm in \cite{TurekWY92}.
	Let $\tau''$ be the first point in time where both jobs scheduled at $\tau'$ have ended.
	Note that after this point in time the number of idle machines is monotonically increasing per time step. 
	Hence, we can use the improved List-Schedule algorithm without constructing any machine conflicts.
	
	To analyze the approximation ratio after adding the residual jobs, let $\rho$ be the last point in the schedule where a job is started. 
	If this job has a width of at most $m/3$ at every time before $\rho$ and after $T$ the number of idle machines is at most $m/3$ since otherwise this job would have been started earlier. 
	If this job has a machine requirement larger than $m/3$ it has been started before $T$. 
	In both cases the total idle time up to  $\rho$ is bounded by $\rho m/3$. 
	As a consequence, we have $\rho \leq 3/2 \cdot \OPT$ since all jobs start before $\rho$ and $m\OPT \geq \sum_{j \in \jobs} p_jq_j\geq (2/3)\rho m$. 
	Therefore, the schedule has a makespan of at most $(2/3)\OPT +p_{\max}$.
	\end{caseList}
	
	We have proven that in both cases $a>b$ and $a\leq b$ the described algorithm produces a schedule with makespan at most $(2/3)\OPT +p_{\max}$.
	This algorithm has a running time of the form $\Oh(n \log(n))$: The sorting of the items is possible in $\Oh(n \log(n))$; each of the values $a$, $b$ and $\tau''$ can be found in $\Oh(n)$; and last the optimized List-Schedule can be implemented to be in $\Oh(n \log(n))$ by organizing the relevant points in time as well as the set of items inside a search tree.
	
	Last, we describe how to partition this schedule into the schedule on the two clusters $C_1$ and $C_2$ as needed for the algorithm in Lemma \ref{lma:partitioningLemma}.
	Note that in all the described cases the additional $p_{\max}$ is added by the last started job. 
	To partition this schedule such that it is scheduled on the two clusters $C_1$ and $C_2$, we look at the starting time $\rho$ of the last started job. 
	We remove this last started job and all the jobs which end strictly after $\rho$ and place them into the second cluster $C_2$ and leave the rest untouched to be the schedule for $C_1$. 
	As we noted before the schedule up to $\rho$ has a height of at most $(3/2)\OPT$.
	Furthermore, since the last job starts at $\rho$, all the removed jobs have a total machine requirement of at most $m$, and, hence, we can start them all at the same time. 
	The resulting schedule on $C_2$ has a height of at most $p_{\max}$.
\end{proof}

%\begin{pcvstack}[center]
%	\procedure[linenumbering, mode=text] {FastApprox($\jobs$, $m$)}{%
%		Sort job in $\jobs$ by width\\
%		Partition $\jobs$ into $\jobs_{[m/2,m]}$, $\jobs_{[m/3,m/2)}$ and $\jobs_{(0,m/3)}$\\
%		Stack jobs in $\jobs_{[m/2,m]}$ in order of width on top of each other (widest first)\\
%		Define list $L$ of endpoints of jobs in $\jobs_{[m/2,m]}$ and number of idle machines\\
%		$T:=0$\\
%		\pcwhile $T < \pT{\jobs_{[m/2,m]}}$ \pcdo\\
%		\pcind search for the widest job $j$ in $\jobs_{[m/3,m/2)}$ that fits at $T$\\
%		\pcind \pcif $j$ exists \pcdo\\
%		\pcind \pcind place $j$ at $T$, $T = T + \pT{j}$\\
%		\pcind \pcelse\\
%		\pcind \pcind Set $T$ to next larger element in $L$\\
%		Continue scheduling the widest jobs in $\jobs_{[m/3,m/2)}$ at the next free position
%		Calculate $a$ and $b$
% 	}
%\end{pcvstack}

In the next step, we present the technique to divide this schedule on $C_1$ and $C_2$ to $N$ clusters and prove Theorem \ref{thm:fastApproximationMSCP} in this way. 
The used technique is similar to the technique in Section \ref{sec:Partitioning}.
However, it is no longer possible to partition the schedule into sections of height at most $2\OPT$.

\subsection{Proof of Theorem \ref{thm:fastApproximationMSCP}}
In this section, we prove Theorem \ref{thm:fastApproximationMSCP}.
We start with the schedule given by the $(3/2)\OPT +p_{\max}$ algorithm from Lemma~\ref{lma:fastPTS} and its partition onto the two clusters $C_1$ and $C_2$.
To partition the schedule on $C_1$ onto the different clusters, we differentiate the three cases $N = 3i$, $N = 3i+1$ and $N=3i+2$. 

%\textit{The case $N = 3i$:} 
\begin{caseList}
%\subparagraph*{Case 1: $N = 3i$}
\item[$N = 3i$]
In this case, the schedule on $C_1$ has a height of $T \leq (9i/2)\OPT$.
We partition it into $2i$ parts of equal height $T/(2i) \leq (9/4)\OPT$. 
During this partition step, we cut the schedule $2i-1$ times.
The jobs intersected by this cut have to be scheduled separately using height $p_{max}$. 
Together with the jobs in $C_2$, we have $2i$ sets of jobs with height bounded by $p_{max}$ and machine requirement bounded by $m$.
We schedule these sets pairwise in $i$ additional clusters analogously to the clusters of type B in Section \ref{sec:Partitioning}. 
In total, we use $3i = N$ Clusters and the largest one a has height of at most $(9/4) \OPT = 2.25\OPT$.

%\textit{The case $N = 3i +1$:} 
%\subparagraph*{Case 2: $N = 3i+1$}
\item[$N = 3i +1$]
In this case, the schedule on $C_1$ has a height of $T \leq (3(3i+1)/2)\OPT = ((9i+3)/2)\OPT$.
We partition the schedule into $2i$ parts of equal height and one part with a smaller height. 
On this part, we schedule the jobs from $C_2$ as well.
Let $\Ts := (2/(9i+3))T \leq \OPT$.
The $2i$ parts of equal height have a size of $((9i+5)/(4i+2))\Ts$ and the last part has a height of $((5i+3)/(4i+2))\Ts$. 
It is easy to verify the $2i \cdot ((9i+5)/(4i+2))\Ts + ((5i+3)/(4i+2)) \Ts = ((9i+3)/2) \Ts = T$ and hence we have partitioned the complete schedule on $C_1$. 
By partitioning the schedule on $C_1$ into these parts, we have cut the schedule $2i$ times. 
Therefore, together with the jobs on $C_2$, we have to schedule $2i+1$ parts of height $p_{\max}$. 
We schedule $C_2$ on the cluster with current makespan $((5i+3)/(4i+2))\Ts$ resulting in a schedule of height $((5i+3)/(4i+2))\Ts +p_{\max} \leq ((9i+5)/(4i+2))\OPT$, (since $p_{\max} \leq \OPT$).
We pair the other $2i$ parts and schedule them on $i$ distinct clusters. 
In total, we generate $2i+1+i = 3i+1$ cluster and the largest occurring makespan is bounded by $((9i+5)/(4i+2))\OPT$.

%\textit{The case $N = 3i+2$:} 
%\subparagraph*{Case 1: $N = 3i+2$}
\item[$N = 3i +2$]
In this case, the schedule on $C_1$ has a height of $T \leq (3(3i+2)/2)\OPT = ((9i+6)/2)\OPT$.
Again, we partition this schedule into $2i+1$ parts of equal height and one part with a smaller height. 
On top of this part, we will schedule two parts with processing time $p_{\max}$.
Let $\Ts := (2/(9i+6))T \leq \OPT$.
The first $2i+1$ parts of $C_1$ have a height of $((9i+10)/(4i+4))\Ts$ and the last part has a height of at most $((i+2)/(4i+4))\Ts$. 
It is easy to verify that $(2i+1)((9i+10)/(4i+4))\Ts + ((i + 2)/(4i+4))\Ts = ((9i+6)/2)\Ts = T$ and, hence, we have scheduled all parts of $C_1$. 
Since $((i+2)/(4i+4))\Ts +2p_{\max} \leq ((9i+10)/(4i+4))\OPT$, we can schedule two parts with processing time at most $p_{\max}$ on this cluster. 
We have cut the schedule on $C_1$ exactly $2i+1$ times. 
Together with the jobs from $C_2$, we have $2i+2$ parts with processing time at most $p_{\max}$ we have to schedule inside the other clusters. 
Since we already have scheduled two of these parts, we pair the residual $2i$ parts and generate $i$ new clusters with makespan at most $2p_{\max} \leq 2\OPT$.
In total, we generated $2i+2 +i= 3i+2$ clusters and the largest makespan occurring on the clusters is bounded by $((9i+10)/(4i+4))\OPT$.
\end{caseList}

For each of the three cases $N = 3i$, $N = 3i+1$, and $N = 3i+2$, we have presented a partitioning strategy which partitions the schedule from clusters $C_1$ and $C_2$ onto $N$ clusters such that each cluster has a makespan of at most $(9/4) \OPT$, $((9i+5)/(4i+2))\OPT$ or $((9i+10)/(4i+4))\OPT$ respectively. 
Hence, we have proven Theorem \ref{thm:fastApproximationMSCP}. 
%\end{proof}

\section{An AEPTAS for \acl{SP}}
\label{sec:StripPacking}

In this section, we present an $(1+\eps)\OPT +h_{\max}$ algorithm for \acf{SP} with running time $\mathcal{O}(n/\eps) + \mathcal{O}_{\eps}(1)$ proving Theorem \ref{thm:SP}.
It is inspired by the algorithm in \cite{JansenS09}. 
However, we made some improvements to guarantee an efficient running time.
In the description of this algorithm, we will assume that $1/\eps \in \NN$. 
 
This algorithm combined with the techniques from Section \ref{sec:Partitioning} delivers a 2-approximation algorithm for \acf{MSP}. 
To prove Theorem \ref{thm:MSP}, we need to place some of the jobs on another strip named $C_2$, which has a width of at most $\gamma W$. 
We have either $\gamma = 1$ for the case $N \geq 3$ or $\gamma = 1/8$ for the case that $N = 2$.
In the following description of the algorithm, we proof that the total width of the items placed on top of the packing can be bounded by $\gamma W$, and hence it is possible to place them inside the extra cluster instead of at the top of the packing.
When interested solely in the AEPTAS the value $\gamma$ can be set to $1$.

To find the algorithm for Theorem~\ref{thm:AEPTASMCS} for the \ac{MSP} case, we call the algorithm from this section with $\gamma = 1$ and cut the resulting schedule with makespan $T \leq (1+\eps)N\OPT +h_{\max}$ into parts $N$ of height $(T - h_{\max})/N$, such that the items overlapping the cut stick together with the part below the cut.
As a result each part has a height of at most $(T - h_{\max})/N + h_{\max} \leq (1+\eps)\OPT + h_{\max}$.

\subsection{Simplify}

Similar as in Section \ref{sec:scheduling}, we start with defining an upper and a lower bound for the approximation ratio. 
Let $\areaI{\items} := \sum_{i \in \items}\iW{i}\iH{i}$ be the total area of all the items and let $h_{\max}$ be the largest occurring height in $\items$. 
By Steinberg \cite{Steinberg97}, we now that $\max\{\areaI{\items},h_{\max}\} \leq \OPT \leq 2\max\{\areaI{\items},h_{\max}\}$ and we define $T := \max\{\areaI{\items},h_{\max}\}$.

In the first step, we partition the items by their size. 
Other than in the algorithm for \acf{PTS}, we need a gap between wide and narrow items as well. 
Hence, we partition the items into 
large $\itemsL :=\{i \in \items| \iH{i} \geq \delta T, \iW{i}\geq \delta W\}$, 
vertical $\itemsV :=\{i \in \items| \iH{i} \geq \delta T, \iW{i} \leq \mu W\}$, 
horizontal $\itemsH :=\{i \in \items| \iH{i} \leq \mu T, \iW{i} \geq \delta W\}$, 
small $\itemsS :=\{i \in \items| \iH{i} \leq \mu T, \iW{i} \leq \mu W\}$ 
and medium sized items $\itemsM :=\items \setminus (\itemsL \cup \itemsV \cup \itemsH \cup \itemsS)$ for some $\delta,\mu \leq \eps$, see Figure \ref{fig:partitoningOfTheItems}.

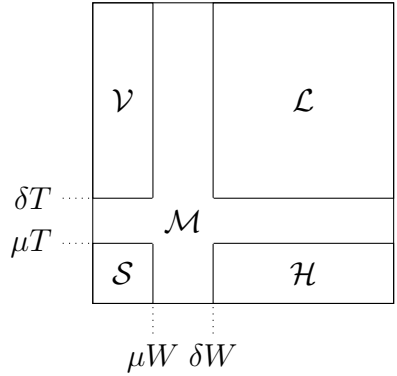
\begin{figure}[ht]
	\centering
	\begin{tikzpicture}
	\pgfmathsetmacro{\w}{4}
	\pgfmathsetmacro{\h}{4}
	
	\pgfmathsetmacro{\dW}{0.4}
	\pgfmathsetmacro{\mW}{0.2}
	
	\pgfmathsetmacro{\dT}{0.35}
	\pgfmathsetmacro{\mT}{0.2}
	
	\draw (0,0) rectangle (\w,\h);
	\draw[dotted] (\dW*\w,-0.1*\h) node[below]{$\delta W$} -- (\dW*\w,0*\h); 
	\draw[dotted] (\mW*\w,-0.1*\h) node[below]{$\mu W$} -- (\mW*\w,0*\h); 
	
	\draw[dotted] (-0.1*\w,\dT*\h) node[left]{$\delta T$} -- (0*\w,\dT*\h); 
	\draw[dotted] (-0.1*\w,\mT*\h) node[left]{$\mu T$} -- (0*\w,\mT*\h); 
	
	\draw (0, \dT *\h) rectangle node[midway] {$\itemsV$}(\mW *\w,1 *\h);
	\draw (\dW*\w, \dT *\h) rectangle node[midway] {$\itemsL$}(1 *\w,1 *\h);
	
	\draw (0*\w, 0 *\h) rectangle node[midway] {$\itemsS$}(\mW *\w,\mT *\h);
	
	\draw (\dW *\w,0 *\h) rectangle node[midway] {$\itemsH$}(1 *\w,\mT *\h);
	\draw[white] (\mW *\w,\mT *\h) rectangle node[midway,black] {$\itemsM$}(\dW *\w,\dT *\h);
	\end{tikzpicture}
	\caption{This Figure shows the partition of the items. Each item $i \in \items$ can be represented by a point in this two-dimensional plane. The x-coordinate of the point corresponds to the width of the item, while the y-coordinate corresponds to the height of the item.}
	\label{fig:partitoningOfTheItems}
\end{figure}

We will discard the medium sized items and place them at the end of the packing. 
To make this possible the total area of the medium sized items has to be small. 
In the next lemma, we show that we can find values for $\delta$ and $\mu$ which guarantee this property.

\begin{lemma}
	\label{lma:StripPackingGapMediumitems}
	Consider the sequence $\sigma_0 := \eps $, $\sigma_{i+1} = \sigma_{i+1}^3\eps^5\gamma/x$. 
	There exists an $j \in \{0,\dots,2/(\eps\gamma)-1\}$ such that when defining $\delta = \sigma_j$ and $\mu = \sigma_{j+1}$ the total area of the medium sized items $\items_M$ is bounded by $\gamma \eps W T$.
\end{lemma}
\begin{proof}
	This Lemma follows by a direct application of the pigeon hole principle. 
	Let $\itemsM_{j}$ be the set of medium sized items when defining $\delta = \sigma_j$ and $\mu = \sigma_{j+1}$.
	Each item $i \in \items$ can appear in at most two of these sets, in the first because its width is between $\mu W$ and $\delta W$ and in the second, because its height is between $\mu T$ and $\delta T$.
	Assume that all the sets $\itemsM_{j}$ have an area $\areaI{\itemsM_{j}} > \gamma \eps \cdot W\cdot T$.
	As a consequence the total area of all these sets is at least $\sum_{j = 0}^{2/(\eps\gamma)-1}\areaI{\itemsM_{j}} > 2 \cdot W\cdot T$, a contradiction since the total area of all the items is bounded by $W\cdot T$.
\end{proof}
Furthermore, it holds that $\delta \geq (\eps\gamma/x)^{3^{\mathcal{O}(1/(\eps\gamma))}}$. 
We define $\delta ':= \eps^k$ as the maximum number such that $\eps \delta' \geq \delta \geq \delta'$.
Note that $k \in 3^{\mathcal{O}(1/(\eps\gamma))}$ and use $\delta'$ for the partitioning of the items.
As a consequence, the the area of the medium items is still at most $\gamma \eps W T$, but the distance between $\delta'$ and $\mu$ is reduced, i.e. we have $\mu = \delta^3\eps^5 \gamma/x \leq (\delta'/\eps)^3\eps^5 \gamma/x = \delta'^3\eps^2 \gamma/x$. 
However, for simplicity of notation, we will write $\delta$ instead of $\delta'$ in the following and use $\mu \leq \delta\eps^2 \gamma/x$ respectively.

In the second step, we round the heights of the items. 
By increasing the packing height by at most $\eps T$, we can round the heights of the items to multiples of $\eps T/n$, because adding $\eps T/n$ to each processing time lengthens the packing by at most $n \cdot \eps T/n$.
Hence after this rounding step, we have $T \leq \OPT \leq (2+\eps)T$.
Since each item has a height of at most $T$, there are at most $n/\eps$ different item sizes, and hence, sorting them by height can be done in $\mathcal{O}(n/\eps)$ using Bucket-Sort.
Furthermore, the largest $l$ such that $p_j \in \{\eps^l T,\eps^{l-1}T\}$ is bounded by $\mathcal{O}(\log(n))$.

In the next step, we scale the instance with $n/\eps T$. 
As a result all the items have a height that is one of the integral values $\{1,2,\dots, n/\eps\}$ and the optimal packing height for this scaled instance is one of the integral values $\{n/\eps,n/\eps +1, \dots,2n/\eps +n\}$, because for the rounded instance it holds that $T \leq \OPT \leq 2T + \eps T$ and the optimal packing height has to be integral since all the item heights are integral.
We scale $T$ accordingly such that $T = n/\eps$.
In the algorithm, we will do a binary search over the packing heights. 

In the next step, we use the same geometric rounding as above to round the heights of the items to fewer different sizes using Lemma \ref{lma:roundingProcessingTimes} and loose a factor of at most $(1+2\eps)$ in the approximation ratio with regard to the scaled instance.
Now the items have at most $\mathcal{O}(\min\{n/\eps,\log(n)/\eps^2\})$ possible different sizes and, without any further loss, we can assume that all large and vertical items start at multiples of $\eps\delta'T$. 
We call the area between two consecutive multiples of $\eps\delta'T$ a layer and number them starting at zero.
To ensue the integrity of the item heights, we scale the instance with $1/(\eps \delta)$ before the rounding step and scale $T$ accordingly such that $T = n/(\eps^2\delta)$.
Note that $1/(\eps\delta) \in \NN$ since $1/\eps \in \NN.$
To this point, we know that with out all the scaling steps it holds that $T \leq \OPT \leq (1+2\eps)(2+\eps)T$.
Hence the number of layers $L$ in an optimal solution is at least $1/(\eps\delta)$ and at most $(1+2\eps)(2+\eps)/(\eps\delta) \leq 5/(\eps\delta)$ for $\eps \leq 1/2$.
%Note that to place the horizontal items we can use the original heights the rounding is needed for all the other items only.

In the next step, we remove all small and medium sized items from the optimal packing and use a Lemma from \cite{JansenRau17} which states that we can partition any optimal packing into a constant number of sub areas, such that each subarea contains just one type of item.

\begin{lemma}[See \cite{JansenR16}]
\label{lma:roughPartition}
We can partition the area $W \times (1 +2\eps)\OPT$ into $ \mathcal{O}(1/(\eps\delta'^2))$ rectangular areas called boxes. 
\begin{itemize}
\item Each large item $i\in \itemsL$ is contained in its personal box of height $\iH{i}$ and width $\iW{i}$.
%\item There are at most $(1+2\eps)/(\eps\delta^2) - |L|/\delta$ many boxes containing horizontal items $i \in \items_H$. Each of them has a height of $\eps\delta T$ and a width larger than $\delta W$.
\item There are at most $\Oh(1/(\eps\delta^2))$ many boxes containing horizontal items $i \in \itemsH$. 
Each of them has a height of $\eps\delta' T$ and a width larger than $\delta' W$.
%\item There are at most $3(1+2\eps)/(\eps\delta^2)$ many boxes containing vertical items $i \in \items_V$.
\item There are at most $\mathcal{O}(1/(\eps\delta^2))$ many boxes containing vertical items $i \in \itemsV$.
\item No item in $\itemsH$ is intersected vertically by any box border, but can be intersected horizontally
\item No item in $\itemsV$ is intersected horizontally by any box border, but can be intersected vertically.
\item Each boxes lower and upper borders are at multiples of $1/(\eps\delta)\OPT$
\end{itemize}
\end{lemma}

In the algorithm, we cannot try each of these partitions since then the width of the strip $W$ would appear linear in the running time.
Instead, we are interested in the relative positioning of the large items and the boxes for horizontal items.

\subsection{Boxes for horizontal rectangles}

The last simplification step is the rounding of widths of the horizontal items. 
We call the set of generated rounded items $\bar{\itemsH}$.

\begin{lemma}
	\label{lma:RoundingHorizontalItems}
We can round the width of the horizontal items to $\mathcal{O}(\log(1/\delta)/\eps)$ different sizes in at most $\mathcal{O}(n \log(1/\eps))$ operations. 
These rounded items can be placed fractionally instead of the horizontal items and an extra box of height $\eps T$.
\end{lemma}
\begin{proof}
To round the width of the items, we use a similar technique as for rounding the machine requirements of the small wide jobs in Lemma \ref{lma:fastRounding} called geometric grouping. 
This technique was first introduced by \cite{VegaL81} as well.
The difference to linear grouping is an additional partitioning step prior to the steps of the linear grouping, as described below.

We first partition the set of horizontal items into the following $\mathcal{O}(\log(1/\delta))$ sets $\items_{H,i} := \sset{i \in \items_H}{ W/2^{i+1}<w_i\leq W/2^i}$. 
For each of these sets, we perform the steps of linear grouping with a customized adjustment to the height of the segments per set.
For these adjusted heights, we use the fact that it is possible to place at least $2^{i}$ and at most $2^{i+1}$ items from the set $\items_{H,i}$ next to each other into the strip.

For each $\items_{H,i}$, we stack the contained items in order of decreasing width and partition this stack into $1/\eps$ segments of size $\eps h(\items_{H,i})$, where $h(\items_{H,i})$ is the total height of the items in $\items_{H,i}$ using the original item heights. 
We define a new job for each segment which has height $\eps h(\items_{H,i})$ and width of the widest item intersecting this segment, see Figure \ref{fig:LinearGrouping}.
The widest item will be placed at the end of the schedule inside a new box. 
Since we are allowed to place this item fractionally and we can place at least $2^i$ of these fractions next to each other, we need at most $(\eps h(\items_{H,i}))/2^i = \eps h(\items_{H,i})/2^{i}$ additional height to place this item.

To place all the largest rounded items from each set $\items_{H,i}$, we introduce a new box for horizontal items.
We define the boxes height as $\sum_{i = 0}^{\mathcal{O}(1/\log(\delta))}\eps h(\items_{H,i})/2^{i}$.
For each $i \in \mathbb{N}$, the total width of $2^{i+1}$ items from the set $\items_{H,i}$ is larger than $W$
and hence $2 \sum_{i = 0}^{\mathcal{O}(\log(1/\delta))}h(\items_{H,i})/2^{i+1} \leq 2T$ since $ T \cdot W \leq 2 \area{\items}$. 
Therefore, the height of the introduced box is bounded by $\eps T$.
The total number of different item widths is bounded by $\mathcal{O}(\log(1/\delta)/\eps)$.%$\mathcal{O}(2^{x/\eps}\log(1/\eps)/\eps)$.

Regrading the running time, as seen above in the proof of Lemma \ref{lma:fastRounding}, the size defining items can be found in $\mathcal{O}(|\items_{H,i}|\log(1/\eps))$ for each set $\items_{H,i}$. 
Therefore, all the sizes can be found in $\mathcal{O}(\sum_{i = 0}^{\mathcal{O}(2^{x/\eps}\log(1/\eps))}|\items_{H,i}|\log(1/\eps))= \mathcal{O}(n \log(1/\eps))$.
\end{proof}

In the next step, we show that it is possible to reduce the number of widths for horizontal boxes to be constant depending on $\delta$. 
We do this in order to make it possible for the algorithm to guess their sizes in polynomial time.

\begin{lemma}
Given a partition of the optimal solution into boxes, we can reduce the number of possible width for the boxes to $|\bar{\items}_H|^{1/\delta}$ and guarantee that at most $\Oh(1/(\eps\delta))$ of these sizes are used in the partition by exactly $1/\delta$ boxes each.
This rounding step adds at most $\eps T$ to the packing height.
\end{lemma}

\begin{proof}
%(TODO: no small items in the boxes)
We reduce the number of box sizes in two steps. 
First, we reduce the possible number of box sizes, by shrinking the boxes to be a combination of widths of the rounded horizontal items.
In the second step, we reduce the number of different box sizes per solution by using a linear grouping step.

Look at one box $B$ for horizontal items.
We can shift all the horizontal items in this box to the left as much as possible such that all the left borders of the horizontal items are touching either the box border or the right side of another horizontal item.
If the left border of the box does not touch the leftmost item, we can move this border to the left until it does. 
Now the box for horizontal items has a width which is the sum of widths of rounded horizontal items, i.e. $w(B)\in \{\sum_{i = 1}^{1/\delta-1}w_i|i \in \bar{\items}_H\}$.
%(TODO: rounded horizontal items).
As a result the total number of possible box widths is bounded by $|\bar{\items}_H|^{1/\delta}$.
% Todo: kann dies besser abgeschätzt werden, da tatsächlich nur $1/\eps$ größen zur verfügung stehen um tatsächlich 1/\delta items hintereinander zu packen? -> ist nicht hilfreich, da der exponent noch immer 1/\delta bleibt.

Given such a set of boxes, we can use linear grouping to reduce the total number of different box widths. 
Since the optimal packing has a height of at most $(1+2\eps)(1+\eps)2T$ and each box has a height of $\eps\delta T$ and there are at most $1/\delta$ boxes for horizontal items in each layer,
a sorted stack of all the boxes has a total height of at most $\eps\delta T \cdot (1+2\eps)(1 +\eps)2/\eps\delta^2 \leq (1+2\eps)(1+\eps)2T /\delta$.
We partition the set of boxes such that the the $1/\delta$ widest boxes are contained in the first set, the $1/\delta$ next most wide boxes are contained in the second set and so on.
As a result, the total height of each set of boxes is bounded by $\eps T$ and the set of boxes is partitioned into at most $(1+2\eps)(1+\eps)2/(\eps\delta) = \Oh(1/(\eps\delta))$ groups. 
Note that the last group might contain less boxes than $1/\delta$. 
To enforce that after the rounding there are $1/\delta$ boxes of each width, we assume that the last group has additional boxes with width zero.
We round the box widths to the largest box width of the corresponding set.
Again the last rounded group of boxes has to be positioned at the end of the packing adding at most $\eps T$ to the packing height.
% geometic grouping kann hier eventuell noch helfen und die anzahl der rateschritte veringern.
\end{proof}

Let $\mathcal{W}_B$ be the set of rounded widths of the boxes.
Note that $\mathcal{W}_B$ can contain less than $2(1+2\eps)(1+\eps)/(\eps\delta)$ sizes if there are less than $2(1+2\eps)(1+\eps)/(\eps\delta^2) -1/(2\delta)$ boxes in the partition of the optimal instance.
To place the horizontal items, we first guess the set $\mathcal{W}_B$. 
There are at most $\mathcal{O}((|\bar{\items}_H|^{1/\delta})^{\Oh(1/(\eps\delta))}) \leq \mathcal{O}((\log(1/\delta)/\eps)^{\Oh(1/(\eps\delta^2))})$ possibilities for this set.
%Note that for each guessed width, there are exactly $1/\delta$ boxes of this width.

After we guessed the set of boxes, we check with a linear program whether all the rounded horizontal items can be placed into the boxes. 
Similar to the placing of small jobs in Section \ref{sec:smallJobs}, we use configurations to place the horizontal items into the boxes. 
A configuration of horizontal items is a multiset $C :=\{a_{i,C}: i |i \in \bar{\items}_H\}$. 
Let $\mathcal{C}$ be the set of all configurations.
We say a configuration $C$ has width $w(C) := \sum_{i \in \bar{\items}_H}a_{i,C}w(i)$. 
Let $\mathcal{C}_w$ be the set of configurations with width at most $w$, i.e.,  $\mathcal{C}_w := \{C \in \mathcal{C} | w(C) \leq w\}$.

Consider the following linear program $LP_{small}$.
\begin{align}
\sum_{C \in \mathcal{C}_{W}}x_{C,W} & = \eps T & \label{eq:confLPStripPackingExtraBox}\\
\sum_{C \in \mathcal{C}_{w}}x_{C,w} & = \eps \delta T & \forall w \in \mathcal{W}_B \label{eq:confLPStripPackingConfigConstraint}\\
\sum_{l = 1}^{|S|} \sum_{C \in \mathcal{C}_{m_l}} x_{C,l}a_{j,C} & = h_j & \forall j \in \bar{\items}_H \label{eq:confLPStripPackingItemConstraint}\\
%\sum_{l = 1}^{|S|} \sum_{C \in \mathcal{C}_{m_l}} x_{C}^{(l)}(m_l -q(C)) & \geq W(\jobs_{S,N}) \\
x_{C,l} &\geq 0  & \forall l = 1, \dots, |S|, C \in \mathcal{C}_{w}
\end{align}
The variables $x_{C,w}$ represent the height of a configuration $C$ inside the boxes of width $w$. 
The sum of these heights should equal the total height of the boxes having this width, which is ensured by the equation (\ref{eq:confLPStripPackingConfigConstraint}).
Equation (\ref{eq:confLPStripPackingExtraBox}) is introduced to represent the extra box for the horizontal items we need due to the rounding of these items.
In the other hand each horizontal item should be covered by the configurations, which is ensured by the equation (\ref{eq:confLPStripPackingItemConstraint}).

Similar as for placing the small narrow jobs in Section \ref{sec:smallJobs}, we solve a relaxed version of this linear program called $LP_{small,rel}$. 
In this relaxed version, we replace equation (\ref{eq:confLPStripPackingConfigConstraint}) by the equation 
\[\sum_{C \in \mathcal{C}_{w}}x_{C,w} = (1+\eps^2)\eps \delta T \  \forall w \in \mathcal{W}_B \]
and, similarly, we replace the equation (\ref{eq:confLPStripPackingExtraBox}) by 
\[\sum_{C \in \mathcal{C}_{W}}x_{C,W}  = (1+\eps^2)\eps T.\]

\begin{lemma}
	If there is a solution to $LP_{small}$, we can find a basic solution to $LP_{small, rel}$ in $\mathcal{O}((|\mathcal{W}_B||\bar{\items}_H|(\ln(|\bar{\items}_H|) +\eps^{-4}))^{1.5356}(|\mathcal{W}_B|+|\bar{\items}_H| + (\log(1/\eps))^3/\eps^4)) \leq \mathcal{O}(\log(1/\delta)^{1.5356}/\eps^{6}\delta^6)$ operations.
\end{lemma}

\begin{proof}
Note that the described linear program and the described configurations are equivalent to the ones for the small narrow jobs. 
Hence, we can use the algorithm proposed in Lemma \ref{lma:LPSolveHorizontalJobs} to find the desired basic solution.
\end{proof}

We call the set of guessed boxes for horizontal items $\mathcal{B}_H$. 
In the end of the algorithm, we place the configurations inside the boxes and the horizontal items (fractionally) into the configurations similar to the placement of small wide jobs in Section \ref{sec:smallJobs}.
A basic solution of the above linear program has at most $|\mathcal{W}_B| + |\bar{\items}_H| +1$ non zero components. 
When filling the configurations inside the boxes $\mathcal{B}_H$, we have to cut the configurations at the box borders of boxes with the same size.
Hence inside the boxes, we have at most $|\mathcal{B}_H| + |\bar{\items}_H| +1$ configurations.
At each configuration border, we generate fractionally placed horizontal items.
However these items all fit next to each other since they are inside one configuration. 
Hence, we can remove the cut items and shift them up to the top of the packing. 
This step adds at most $\mu T \cdot (|\mathcal{B}_H| + |\bar{\items}_H| +1) \leq \mu T (\log(1/\delta)/\eps + \mathcal{O}(1/\eps\delta^2)) = \mathcal{O}(\eps T)$ to the packing height.

\subsection{Positioning containers as well as  large and vertical rectangles}
\label{sec:PositioningContainersLargeVerticalItems}
In this section, we handle the positioning of the boxes for horizontal items and the placement of large and vertical items. 
These boxes and items are positioned by guessing the x-coordinate of the lower left corner, which has to be a multiple of $\eps\delta T$.
Afterward, we guess the order from left to right in which these items and boxes will appear.
The technique described in this section is inspired by the techniques described in \cite{JansenS09} Chapter 4.

In the first step, we guess the position of the lower corners of the items and boxes in $\items_L$ and $\mathcal{B}_H$. Note that since the boxes have an area of at least $\eps\delta T \cdot \delta W$ and the large items have an area of at least $\delta T \cdot \delta W$ and the packing has an area of at most $(1+2\eps)(1+\eps)T W$, there are at most $\mathcal{O}(1/(\eps\delta^2))$ boxes and items.
Hence, the total number of possible guesses for positions of their bottom edges is bounded by $(1/\eps\delta)^{\mathcal{O}(1/(\eps\delta^2))}$.

Consider an optimal packing where all the items are rounded and the horizontal items are positioned in the rounded boxes. 
For each large item or box $i \in \items_L \cup \mathcal{B}_H$, we can determine the value of the $y$-coordinates of their left and right borders $y_{i,l}$ and $y_{i,r}$. 
Let $\mathcal{Y}$ be the set of all these $y$-coordinates $y_{i,l}$ and $y_{i,r}$.
We order $\mathcal{Y}$ by value of the coordinates in the optimal packing.
This gives us a permutation $\pi: \mathcal{Y} \to \{1, \dots, |\mathcal{Y}|\}$ from the left and right corners of items and boxes to positions in the ordered list.
Since the value of $W$ is not logarithmically bounded in the input size, we cannot guess the values of the $y$-coordinates in polynomial time.
However, it is possible to guess the correct permutation $\pi$ in $|\mathcal{Y}|! \in (1/(\eps\delta^2))^{\mathcal{O}(1/(\eps\delta^2))}$ guesses.  
For a given item or box $i \in \items_L \cup \mathcal{B}_H$, we write $\pi(i,l)$ to refer to the position of $y_{i,l}$ and analogously $\pi(i,r)$ for the position of $y_{i,r}$ and write $y_j$ to refer to the $y$-coordinate which is mapped to position $j$ in the ordered list.

After these two guesses, the guess of the positions of lower borders and the guess of order of the items, the algorithm tests if this guess was feasible, by testing if it is possible  at all to position the items as forced by this guess.
This can be done in $\mathcal{O}(n)$ by starting with the left most item and position the items one by one in order of the $y$-coordinates as most to the left as possible by the constraints guessed.
As soon as a constraint has to be violated, we stop and discard the guess.
Possible violations of the constraints can be, e.g., that an items left border has to be placed between a left and a right border of another item but this item and the to be placed item overlap  the same horizontal line that an item has to be placed such that it overlaps the right border of the strip that $\pi(i,l) > \pi(i,r)$.  

Consider a feasible guess of starting positions and permutation. 
The next step of the algorithm is to find values for the $y$-coordinates of the left and right borders. 
It determines these values by using a linear program as described below. 
Indeed, since the vertical items have to be placed correctly as well, the linear program is not only concerned about determining the $y$-coordinates, but to place the vertical items as well. 
Consider two consecutive $y$-coordinates $y_j$ and $y_{j+1}$ and the segments of the layers between these. 
Some of them are occupied by an item or a box in $\items_L \cup \mathcal{B}_H$ and some are not. 
We will use the not occupied layers to place the vertical items. 
We scan the area between $y_j$ and $y_{j+1}$ from bottom to top and fuse each set of contiguous unoccupied layers to a box for vertical items. 
Let $\mathcal{B}_{V,j}$ be the set of constructed boxes for the area between the coordinates $y_j$ and $y_{j+1}$. 
Note that there can be at most $\mathcal{O}(1/\eps\delta)$ of them.

Similar as for the horizontal items, we define configurations for the vertical items. 
However instead of placing these items next to each other, we will stack the items inside a configuration for vertical items on top of each other.
Note that in each optimal packing a vertical line through the packing intersects at most $1/\delta$ of these items and hence configurations should contain at most this number of items.
We define a new set of vertical items called $\bar{\items}_V$.
For each appearing item height $h \in \{\iH{i} | i \in \items_V \}$, the set $\bar{\items}_V$ contains one job of height $h$ and width $\sum_{i \in \items_V, \iH{i} = j} \iW{i}$.
To reduce the running time, we will schedule the jobs in the set $\bar{\items}_V$ fractionally instead of the original vertical items. 
Note that $|\bar{\items}_V| \leq \log_{\eps}(1/\delta)/\eps^2$ due to the rounding of the vertical items.

A configuration for vertical items is a multiset $C := \{a_{i,C}: i |i \in \bar{\items}_V\}$ such that $ \sum_{i \in \bar{\items}_V}a_{i,C}\cdot \iH{i} \leq 1/\delta$ and we define its height as $\iH{C} :=\sum_{i \in \bar{\items}_V} a_{i,C} \cdot \iH{i}$.
Let $\mathcal{C}_V$ be the set of all these configurations and let $\mathcal{C}_{V,h}$ be the set of all configurations with height at most $h$.
These configurations for vertical items are combined to hyper configurations which represent the distribution of vertical items in a vertical line through the packing. 
For each segment between two coordinates $y_j$ and $y_{j+1}$, we define a configuration $C_j$ as a tuple of configurations, such that there is exactly one configuration for each of the boxes in $\mathcal{B}_{V,j}$, i.e., $C_j = (C \in \mathcal{C}_{V,\iH{b}}:b \in \mathcal{B}_{V,j})$.
Let $\mathcal{C}_{V,j}$ be the set of all configurations for the section between the coordinates  $y_j$ and $y_{j+1}$.
We define $a_i(C)$ at the number of appearances of item $i \in \bar{\items}_V$ inside the configuration $C \in \mathcal{C}_{V,j}$.
Note that the configurations for the boxes each have a maximum amount of vertical items they can contain and the sum of these numbers is bounded by $1/\delta$.
Hence the total number of different configurations in $\mathcal{C}_{V,j}$ is bounded by $|\bar{\items}_V|^{1/\delta}$.
To find fitting values for the $y$-coordinates the algorithm solves the following linear program:

\begin{align}
y_0 &= 0  \label{eq:LPVerticalCoordinateCostr1}\\
y_{|\mathcal{Y}| +1} &= W \label{eq:LPVerticalCoordinateCostr2}\\
%Y_{j} &\leq Y_{j+1} & \forall j \in \{0,\dots, |\mathcal{Y}|\} \label{eq:LPVerticalCoordinateCostr3}\\
y_{j+1} - y_{j} &= w_j & \forall j \in \{0,\dots, |\mathcal{Y}|\} \label{eq:LPVerticalCoordinateCostr4}\\
y_{\pi(i,r)} - y_{\pi(i,l)} &= \iW{i} & \forall i \in \items_L \cup \mathcal{B}_H \label{eq:LPVerticalCoordinateCostr5}\\
\sum_{C \in \mathcal{C}_{V,j}} x_{C,j}& = w_j & \forall j \in \{0,\dots, |\mathcal{Y}|\} \label{eq:LPVerticalConfigConstr}\\
\sum_{j \in \{0,\dots, |\mathcal{Y}|\}}\sum_{C \in \mathcal{C}_{V,j}} a_i(C) \cdot x_{C,j}& = \iW{i} & \forall i \in \bar{\items}_V \label{eq:LPVerticalWidthItemConstr}\\
w_j &\geq 0 &\forall j \in \{0,\dots, |\mathcal{Y}|+1\} \\
x_{C,j} &\geq 0 & \forall j \in \{0,\dots, |\mathcal{Y}|+1\}, C \in \mathcal{C}_{V,j} \label{eq:LPVerticalVariableX}\\
y_j &\geq 0 &\forall j \in \{0,\dots, |\mathcal{Y}|+1\}
\end{align}

In this linear program there are three types of variables: $x$, $y$ and $w$.
The variables $y_j$ for $j \in \{0,\dots, |\mathcal{Y}|+1\}$ represent the values of the $y$-coordinates of the item and box borders in $\items_L \cup \mathcal{B}_H$, wheres $y_0$ represents the left border of the strip and $y_{|\mathcal{Y}|+1}$ represents the right borer of the strip. 
The variables $w_j$ for $j \in \{0,\dots, |\mathcal{Y}|\}$ represent the distance between the consecutive $y$-coordinates $y_j$ and $y_{j+1}$.
Last, the variables $x_{C,j}$ represent the width of the configuration $C$ in box $b$ which is positioned between $y_j$ and $y_{j+1}$.

The first three constraints (\ref{eq:LPVerticalCoordinateCostr1}) to (\ref{eq:LPVerticalCoordinateCostr4}) ensure that the $y$-coordinates are positioned in the right order and that we use exactly the width of the strip.
Furthermore, the variables $w_j$ for the width between the $y$-coordinates are defined.
The equation (\ref{eq:LPVerticalCoordinateCostr5}) ensures the $y$-coordinates of the items and boxes in $\items_L \cup \mathcal{B}_H$ are positioned such that their distance equals the widths of the corresponding item.
Equations (\ref{eq:LPVerticalConfigConstr}) and (\ref{eq:LPVerticalWidthItemConstr}) ensure that the vertical items are placed correctly.
The first equation ensures that we do not use a to large width for the configurations inside the boxes while the second equation ensures that all the vertical items can be placed.

The total number of constraints is bounded by 
\begin{align*}
&2|\mathcal{Y}| +2 + |\items_L \cup \mathcal{B}_H| + |\bar{\items}_V| \\
&= \mathcal{O}(1/(\eps\delta^2) + \log_{\eps}(1/\delta)/\eps^2) \\
&= \mathcal{O}(1/(\eps\delta^2)),
\end{align*}
While the total number of variables is bounded by 
\begin{align*}
&2|\mathcal{Y}| + 1 + \sum_{j = 0}^{|\mathcal{Y}|+1}|\mathcal{C}_{V,j}|\\
&= \mathcal{O}((1/(\eps\delta^2))((\log_{\eps}(1/\delta)/\eps^2)^{1/\delta}))\\
&= 2^{\mathcal{O}(1/\eps\delta)}
\end{align*}
Furthermore, all appearing values in the linear program are integer, the largest one on the left hand side is bounded by $1/\delta$ while the right hand side is bounded by $W$. 
We can solve this linear program by guessing the right set of at most $\mathcal{O}(1/(\eps\delta^2))$ non-zero components and then solving the corresponding equation system using Gauß-Jordan elimination in $\Oh((2^{\mathcal{O}(1/\eps\delta)})^{\mathcal{O}(1/(\eps\delta^2))}\cdot (1/(\eps\delta^2))^3) = 2^{\Oh(1/(\eps^2\delta^3))}$.

After we have found such a solution, we fix the values for the variables $y_j$ and $w_j$ for each $j \in \{1,\dots,|\mathcal{Y}|+1\}$ and find a basic solution to the linear program consisting just of the equations (\ref{eq:LPVerticalConfigConstr}), (\ref{eq:LPVerticalWidthItemConstr}), and (\ref{eq:LPVerticalVariableX}).
Such a basic solution has at most $|\bar{\items}_V| + |\mathcal{Y}| \in \mathcal{O}(1/(\eps\delta^2))$ non zero components and hence uses at most this number of configurations.

In the very end of the algorithm, these configurations are filled (fractionally) with the rounded vertical items analogously as small wide jobs items are filed into their configurations, see Section \ref{sec:smallJobs}. 
Since each configuration contains at most $1/\delta$ items and we use at most $\mathcal{O}(1/(\eps\delta^2))$ of them, there are at most $\mathcal{O}(1/(\eps\delta^3))$ fractionally placed vertical items which have a total width of at most $\mathcal{O}(\mu W/(\eps\delta^3))$. 
Since $\mu \leq \gamma\delta^3\eps/x$ for a large enough constant $x$, it holds that the total width of the discarded items is smaller than $\gamma W/2$.
These items are placed on top of the packing, adding at most $p_{\max}$ to the packing height, or in the additional container $C_2$.

\subsection{Placing the Small Items}
\label{sec:PlacingTheSmallItems}
Note that the configurations for vertical and horizontal items might be smaller in height or width as the box they are placed inside, i.e., if a configuration $C$ for vertical items is placed in side a box $b \in \mathcal{B}_{V,j}$ there is a box of free area of width $X_{C,b,j}$ and height $h(b) - h(C)$.
We will use this area to place the small items.
The total free area of this kind has to have the size of $\area{\items_S}$, 
since the configurations contain exactly the total area of the corresponding items and
the total area of all items is at most $(1+2\eps)(1+\eps)TW$ while the packing has a height of at least $(1+2\eps)(1+\eps)T$. 

Since we use at most $\mathcal{O}(1/(\eps\delta^3))$ configurations for vertical items and at most $|\mathcal{B}_H| + |\bar{\items}_H| = \mathcal{O}(1/(\eps\delta^2))$ configurations for horizontal items, there are at most $\mathcal{O}(1/(\eps\delta^3))$ boxes for small items.
We call the set of these boxes $\mathcal{B}_S$.

\begin{lemma}
	We can place the small items inside the $\mathcal{O}(1/(\eps\delta^3))$ boxes $\mathcal{B}_S$ and one additional box of width $W$ and height $2 \eps T + \mu T$.
\end{lemma}
\begin{proof}
Remember that the total area of the boxes is at least $\area{\items_S}$.
%In the first step of the algorithm to place the items, we discard all boxes that have a height of less than $\mu T$ and a width of less than $\mu W$. 
%All the discarded boxes have a total area of at most $\mu T W \cdot \mathcal{O}(1/(\eps^2\delta^4))$.
The algorithm first sorts the small items by height in $\mathcal{O}(n + \log(n)/\eps^2)$ time since the small items have at most $\mathcal{O}(\log(n)/\eps^2)$ different sizes. 
Afterward it considers the boxes for the small items $\mathcal{B}_S$ one by one and fills the small items inside them using the NFDH algorithm.
If an item does not fit inside the considered box, because the item is to wide or has a to large height, the algorithm is finished with this box and considers the next.
All the items that cannot be placed inside the boxes $\mathcal{B}_S$ are placed inside the newly introduced box of width $W$ and height $2 \eps T + \mu T$.

Let us consider the boxes next to the configurations and the free area inside them. 
Let $B$ be such a box.
In $B$ there is a free area of at most $\mu W \cdot \iH{B}$ on one side of $B$ since the small items have a width of at most $\mu W$.
Additionally, there can be free area of at most $\mu T \cdot \iW{B}$ on the top of the box since the items have a height of at least $\mu T$. 
Lastly there can be free area between the items. 
However as indirectly shown by Coffman et al. in \cite{CoffmanGJT80} in the proof of Lemma \ref{thm:NFDH}, the free area provoked this way over all the boxes is bounded by $\mu T \cdot W$ since the items have a maximal height of at most $\mu T$ and the boxes have a maximal width of at most $W$.
In total the free area inside the boxes $\mathcal{B}_S$ is bounded by $\mu T W + \mu T\sum_{B \in \mathcal{B}_S} \iW{B} + \mu W \sum_{B \in \mathcal{B}_S} \iH{B} \leq \mu T W \cdot \mathcal{O}(1/(\eps\delta^3))$.
Since it holds that $\mu \leq \eps^2\delta^3/x$ for a suitable large constant $x$, the total area of the non placed small items has to be bounded by $\eps T W$.
Using Lemma \ref{thm:NFDH}, we can place these non placed items with a total height of at most $2 \eps T + \mu T$ inside the extra box.
\end{proof}

\subsection{Packing medium sized items}
\label{sec:PackingMediumSizedItems}
To place the medium sized items, we partition them into two sets, $\items_{M,V}$  which contains all the items taller than $2\eps T$ and $\items_{M,S} := \items_M \setminus \items_{M,V}$. 
Since the total area of the medium sized items is bounded by $\gamma \eps T W$, the total width of the items in $\items_{M,V}$ is bounded by $\gamma W/2$. 
Hence, we can place all these items at the end of the schedule next to the discarded vertical items. 
In total this adds at most $h_{\max}$ to the schedule.

The jobs in $\items_{M,V}$ have a height of at most $2\eps T$ and an area of at most $\eps T W$.
Hence by Lemma \ref{thm:NFDH}, when using the NDFH algorithm to place these items, we add at most $4\eps T$ to the packing height. 

\subsection{Summary of the algorithm}
In the following, we summarize the steps of the algorithm and give a short overview of the running time. An overview of the generated packing can be found in Figure~\ref{fig:packingOverview}.

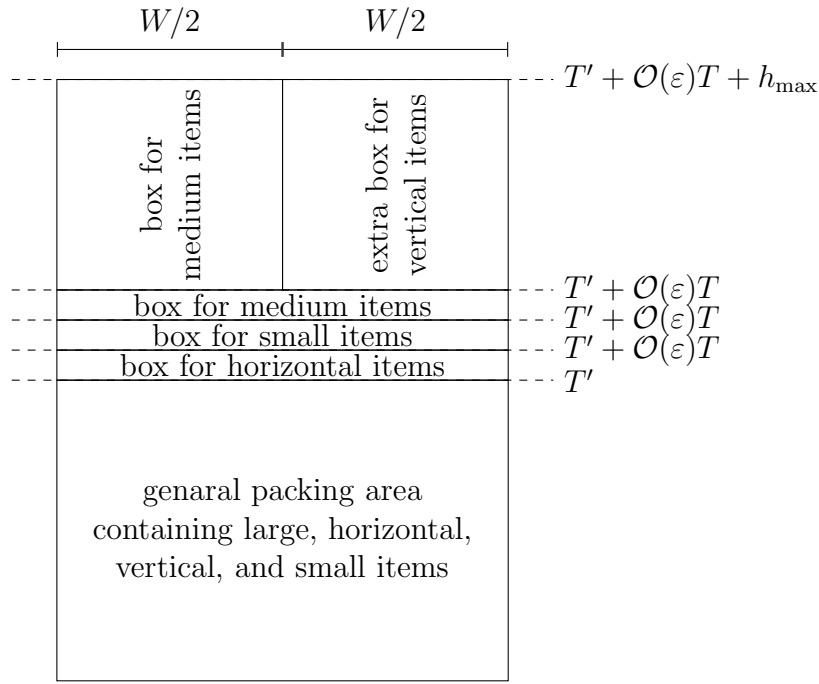
\begin{figure}[ht]
	\centering
	\begin{tikzpicture}
	\pgfmathsetmacro{\h}{4}
	\pgfmathsetmacro{\w}{6}
	
	\foreach \i/\y in {
		$T'$/1,
		$T' + \mathcal{O}(\eps )T$/1.1,
		$T' +\mathcal{O}(\eps )T$/1.2,
		$T' +\mathcal{O}(\eps )T$/1.3,
		$T' +\mathcal{O}(\eps )T + h_{\max}$/2.0
	}{
		\draw[dashed] (-0.1*\w,\y*\h) -- (1.1*\w,\y*\h) node[right]{\i};
	}
	
	\foreach \i/\y/\yy in {
		{genaral packing area\\ containing large, horizontal,\\ vertical, and small items}/0/\h,
		box for horizontal items/\h/1.1*\h,
		box for small items/1.1*\h/1.2*\h,
		box for medium items/1.2*\h/1.3*\h
	}{
		\draw (0,\y) rectangle node[midway,align=center]{\i} (\w,\yy);
	}
	
	\foreach \i/\x/\xx in {
		box for\\medium items/0/0.5*\w,
		extra box for\\vertical items/0.5*\w/\w
	}{
		\draw (\x,1.3*\h) rectangle node[midway, rotate = 90,align=center]{\i} (\xx,2.0*\h);
	}

	\draw[|-|] (0,2.1*\h) -- node[above]{$W/2$} (0.5*\w,2.1*\h);
	\draw[|-|] (0.5*\w,2.1*\h) -- node[above]{$W/2$} (\w,2.1*\h);
	\end{tikzpicture}
	\caption{Overview of the structure of the generated packing }
	\label{fig:packingOverview}
\end{figure}

\begin{enumerate}
	\item \label{enum:SPSimplification} In the first step of the algorithm, we perform the simplification steps.
	We define $T:= \max\{h_{\max}, (\sum\nolimits_{i \in \items}\iH{i}\iW{i})/W\}$, find the correct values for $\delta$ and $\mu$ as described in Lemma \ref{lma:StripPackingGapMediumitems}, and partition the set of items into $\itemsL$, $\itemsV$, $\itemsH$, $\itemsS$, and $\itemsM$ accordingly.
	Afterward, we round the heights and the widths of the items.
	First, we round the height of the items to multiples of $\eps T/n$ and scale the items, such that they have heights in $\{1,\dots, n/\eps\} \subseteq \NN$ and scale $T$ accordingly such that $T = n/\eps$. 
	Next, we scale the instance and $T$ again with $1/\eps\delta$ and use Lemma \ref{lma:roundingProcessingTimes} to round heights of the items in $\itemsL \cup \itemsV$, such that we can assume that they start at multiples of $1/\eps\delta T$.
	Furthermore, introduce the set of rounded items $\bar{\itemsH}$ using Lemma \ref{lma:RoundingHorizontalItems}.
	\item \label{enum:SPBinarySearch} In the next step, we do a binary search over all the possible numbers of layers $L \in [1/(\eps\delta), 5/(\eps\delta)] \cap \NN$.
	Let $T'$ be the currently considered number of layers. 
	For this number of layers, we try to find a packing by performing the following steps. 
	\item \label{enum:SPInsideBinarySearch} For each guess of the set $\mathcal{W}_B$ and each guess of y-coordinates and permutation for boxes and large items: try to solve the configuration linear program $LP_{small}$  to place the horizontal items. If this is not possible try the next guess otherwise try to solve the $LP$ to find the correct positions for the boxes, large items, and vertical items. If this LP is solvable save the guess and LP solutions and try the next smaller value for $T'$ in binary search fashion, otherwise try the next guess. If all guesses fail try the next larger value for $T'$ in binary search fashion.
	\item \label{enum:SPFinish}After use the saved guess and LP solutions to assign the corresponding items.
	First, we revert the scaling of the items and scale the solution and guess accordingly.
	Then, we place the large, vertical, and horizontal items inside the guess as described in Section \ref{sec:PositioningContainersLargeVerticalItems}.
	Afterward, place the small items inside the resulting boxes for small items as described in Section \ref{sec:PlacingTheSmallItems}.
	Finally, we place the medium sized items as described in \ref{sec:PackingMediumSizedItems}.
\end{enumerate}

The step \ref{enum:SPSimplification} takes $\Oh(n\log(1/\eps)+ 1/\eps\gamma)$ operations: 
The set of items needs to be enumerated once to find $T$, i.e., its can be found in $\Oh(n)$. 
The correct values for $\delta$ and $\mu$ can be found in $\Oh(n + 1/\eps\gamma)$ and the corresponding partition can be found in $\Oh(n)$.
The scaling and rounding of the item heights can be done in $\Oh(n)$.
Finally the rounding of the item widths can be done in $\Oh(n\log(1/\eps))$.

The binary search described in Step \ref{enum:SPBinarySearch} can be done in $\Oh(\log(1/(\eps\delta)))$. 
For each of the values given by the binary search framework, there are at most $\mathcal{O}((\log(1/\delta)/\eps)^{\Oh(1/(\eps\delta^2))})$ possibilities to guess $\mathcal{W}_B$, at most $(1/\eps\delta)^{\mathcal{O}(1/(\eps\delta^2))}$ possibilities to guess y-coordinates, and at most $(1/(\eps\delta^2))^{\mathcal{O}(1/(\eps\delta^2))}$ possibilities to guess the right permutation for boxes and large items.
The resulting LP can be solved in $2^{\Oh(1/(\eps^2\delta^3))}$.
Therefore the total running time of steps \ref{enum:SPBinarySearch} and \ref{enum:SPInsideBinarySearch} can  be summarized as 
\begin{align*}
\Oh(\log(1/(\eps\delta))) &\cdot \Oh((\log(1/\delta)/\eps)^{\Oh(1/(\eps\delta^2))}) \cdot (1/\eps\delta)^{\Oh(1/(\eps\delta^2))} \cdot (1/(\eps\delta^2))^{\Oh(1/(\eps\delta^2))} \cdot 2^{\Oh(1/(\eps^2\delta^3))} \\
&\leq 2^{\Oh(1/\eps^2\delta^3)}
\end{align*}

In the final step, we place the original items inside the packing. 
The placement of large, vertical, and horizontal items can be done in $\Oh(n + 1/(\eps\delta^3))$ since there are at most $1/(\eps\delta^3)$ places for vertical and horizontal items.
To place the small items, we use the NFDH algorithm and hence have a running time of at most $\Oh(1/(\eps\delta^3) + n + \log(n)/\eps^2)$ since the items have at most $\log(n)/\eps^2$ sizes and are placed inside at most $\Oh(1/(\eps\delta^3))$ boxes.
The medium sized items can be placed in at most $\Oh(n + 1/\eps^2)$ since they have at most $\Oh(1/\eps^2)$ (possible) different sizes.
Hence the total running time of the algorithm is bounded by $\Oh(n\log(1/\eps)+ 1/\eps\gamma + 2^{\Oh(1/\eps^2\delta^3)} + 1/(\eps\delta^3) + n + \log(n)/\eps^2) \leq \Oh(n\log(1/\eps)+ \log(n)/\eps^2) + 2^{1/(\eps\gamma)^{3^{\Oh(1/(\eps\gamma))}}}$.

As a consequence, we end up with a running time of $\Oh(n\log(1/\eps)+ \log(n)/\eps^2) + 2^{(1/\eps)^{3^{\Oh(1/(\eps))}}}$ for the AEPTAS and wehen using it as a subroutine for \ac{MSP} for $N=3$ because we can choose $\gamma = 1$ in these cases.
On the other hand, when using this algorithm as a subroutine for \ac{MSP} for $N=2$, we end up with a running time of $\Oh(n\log(1/\eps)+ \log(n)/\eps^2) + 2^{(1/\eps)^{3^{\Oh(1/\eps^2)}}}$ because we have to choose $\gamma = \eps$ in this case.

\section{Conclusion}
In this paper, we presented an algorithm for \acf{MCS} and \acf{MSP} with best possible absolute approximation ratio of $2$ and best possible running time $\Oh(n)$ for the case $N \geq 3$.
Still open remains the question if for the case $N=2$ the running time of $\Oh(n\log(n))$ or $\Oh(n\log^2(n)/(\log(\log(n)))$ for \ac{MCS} and \ac{MSP} respectively can be improved to $\Oh(n)$.

Furthermore, we presented a truly fast algorithm for \acf{MCS} with running time $\Oh(n \log(n))$ that does not have any hidden constants.
Since the running time of the $\Oh(n)$ algorithm hides large constants, it would be interesting to improve the running time of the underlying $AEPTAS$ or even to find a faster asymptotic algorithm with approximation guarantee $(5/4)\OPT + p_{\max}$.

\bibliography{lowerBound.bib}

%\begin{appendix}
%\input{appendix.tex}
%\end{appendix}
\end{document}